\documentclass[a4paper,11pt,reqno]{amsart}
\usepackage[a4paper]{geometry}
\usepackage{amssymb,amsmath}
\usepackage{mathtools}
\usepackage{enumerate}
\usepackage{bookmark}
\usepackage{hyperref}
\usepackage[initials,backrefs]{amsrefs}
\numberwithin{equation}{section}
\theoremstyle{plain}
\newtheorem{theorem}{Theorem}

\newtheorem{lemma}{Lemma}
\newtheorem{corollary}{Corollary}
\theoremstyle{definition}
\newtheorem{definition}{Definition}
\newtheorem*{assi*}{(I) Short-range interaction}
\newtheorem*{assp*}{(P) Log-H\"older continuity condition}
\newtheorem*{dskn*}{$\dskn$}
\newtheorem*{dsknn*}{$\dsknn$}
\theoremstyle{remark}
\newtheorem{remark}{Remark}
\newcommand{\prob}[1]{\DP\left\{#1\right\}}
\newcommand{\esm}[1]{\mathbb{E}\left[\,#1\,\right]}
\newcommand{\Bone}{\mathbf{1}}
\newcommand{\BA}{\mathbf{A}}
\newcommand{\BB}{\mathbf{B}}
\newcommand{\BC}{\mathbf{C}}

\newcommand{\BG}{\mathbf{G}}
\newcommand{\BH}{\mathbf{H}}

\newcommand{\BK}{\mathbf{K}}
\newcommand{\BM}{\mathbf{M}}
\newcommand{\BP}{\mathbf{P}}
\newcommand{\BU}{\mathbf{U}}
\newcommand{\BV}{\mathbf{V}}

\newcommand{\BX}{\mathbf{X}}
\newcommand{\Bzero}{\mathbf{0}}
\newcommand{\CH}{\mathcal{H}}
\newcommand{\CJ}{\mathcal{J}}

\newcommand{\DC}{\mathbb{C}}

\newcommand{\DN}{\mathbb{N}}
\newcommand{\DP}{\mathbb{P}}
\newcommand{\DR}{\mathbb{R}}
\newcommand{\DZ}{\mathbb{Z}}
\newcommand{\BDelta}{\mathbf{\Delta}}
\newcommand{\BPsi}{\mathbf{\Psi}}
\newcommand{\Bx}{\mathbf{x}}
\newcommand{\By}{\mathbf{y}}
\newcommand{\Bz}{\mathbf{z}}
\newcommand{\Bu}{\mathbf{u}}
\newcommand{\Bv}{\mathbf{v}}
\newcommand{\FB}{\mathfrak{B}}

\newcommand{\rA}{\mathrm{A}}
\newcommand{\rB}{\mathrm{B}}

\newcommand{\rR}{\mathrm{R}}
\newcommand{\rS}{\mathrm{S}}
\newcommand{\rT}{\mathrm{T}}
\DeclareMathOperator{\card}{card}
\DeclareMathOperator{\diam}{diam}
\DeclareMathOperator{\dist}{dist}
\DeclareMathOperator{\supp}{supp}

\newcommand{\ee}{\mathrm{e}}

\newcommand{\comp}{\mathrm{c}}
\newcommand{\fui}{\mathrm{FI}}

\newcommand{\pai}{\mathrm{PI}}
\newcommand{\sep}{\mathrm{sep}}

\newcommand{\condI}{\mathbf{(I)}}
\newcommand{\condP}{\mathbf{(P)}}
\newcommand{\dsk}[1]{\mathbf{(DS.}k,#1,N\mathbf{)}}
\newcommand{\dsn}[2]{\mathbf{(DS.}#1#2,$\,N\mathbf{)}$}

\newcommand{\dskn}{\mathbf{(DS.}k,N\mathbf{)}}
\newcommand{\dsknn}{\mathbf{(DS.}k,n,N\mathbf{)}}

\newcommand{\dskonn}{\mathbf{(DS.}k+1,n,N\mathbf{)}}
\newcommand{\tto}[1]{\smash{\mathop{\,\,\,\, \longrightarrow \,\,\,\, }\limits_{#1}}}
\begin{document}
%---------------------------------------------------------------------%
\title[Multi-particle Anderson localization at low energy]
{Localization in the multi-particle tight-binding Anderson model at low energy}
%---------------------------------------------------------------------%
\author[T.~Ekanga]{Tr\'esor Ekanga$^{\ast}$}
%---------------------------------------------------------------------%
\address{$^{\ast}$
Universit\'e Paris Diderot Paris 7, UMR 7586 
Institut de Math\'ematiques de Jussieu Paris rive Gauche (IMJ-PRG) et Centre National de la Recherche Scientifique (CNRS),
Batiment Sophie Germain,
13 rue Albert Einstein,
75013 Paris,
France}
\email{ekanga@math.jussieu.fr}
%---------------------------------------------------------------------%
%---------------------------------------------------------------------%
%---------------------------------------------------------------------%
\date{\today}
%---------------------------------------------------------------------%
\begin{abstract}
We consider the multi-particle Anderson tight-binding model and prove that its lower spectral edge is non-random under some mild assumptions on the inter-particle interaction and the random external potential.
We also adapt to the low energy regime the multi-particle multi-scale analysis initially developed by Chulaevsky and Suhov in the high disorder limit, if the marginal probability distribution of 
the i.i.d. random variables is log-H\"older continuous and obtain the spectral exponential and strong dynamical localization near the bottom of the spectrum. 
\end{abstract}
%---------------------------------------------------------------------%
\maketitle
%---------------------------------------------------------------------%
%:s.1
%---------------------------------------------------------------------%
\section{Introduction}\label{sec:intro}

Multi-particle Anderson localization theory is a relatively recent direction in the spectral theory of random operators.  The general structure of an $N$-particle Hamiltonian of a lattice quantum system with interaction is as follows:
\begin{equation}\label{eq:hamiltonian}
\BH^{(N)}(\omega)=-\BDelta+\sum_{j=1}^NV(x_j,\omega)+\BU, \quad x_1,\dots,x_N\in\DZ^d,
\end{equation}
acting in $\ell^2((\DZ^d)^n)$, where $\BDelta$ is the nearest-neighbor lattice Laplacian in $(\DZ^d)^N\cong\DZ^{Nd}$ and $\BU\colon(\DZ^d)^N\to\DR$ is the potential of inter-particle interaction. The external potential $V\colon\DZ^d\times\Omega\to\DR$ is a random field relative to a probability space $(\Omega,\FB,\DP)$ and  acts on $\ell^2((\DZ^d)^N)$ as a multiplication operator by $V(x_1,\omega)+\cdots V(x_N,\omega)$ for $\Bx=(x_1,\ldots,x_N)\in(\DZ^d)^N$.

The first mathematically rigorous results have been obtained by Chulaevsky and Suhov \cites{CS09a,CS09b} with the help of the Multi-Scale Analysis (MSA) and by Aizenman and Warzel \cites{AW09,AW10} who used the Fractional-Moment Method (FMM). In both cases, it was assumed that the random potential field is i.i.d., and the interaction $\BU$ has finite range. 

In \cite{AW09}, the authors proved the multi-particle spectral and dynamical localization for strongly disordered and weakly interacting systems. Due to technical requirements of the FMM, it was assumed in \cite{AW09} that the distribution of the i.i.d. random variables $\{V(x,\omega):x\in\DZ^d\}$ has a bounded density satisfying an additional technical condition. While in \cite{CS09b}, the multi-particle spectral localization was proven for strongly disordered systems with a H\"older continuous distribution.

In the present paper, we study localization near the bottom of the spectrum of $\BH^{(N)}(\omega)$. This complements the results of \cites{AW09,AW10,CS09a,CS09b}, which were concerned with the high-disorder regime. The recent paper \cite{KN13}  also studies localization for $\BH^{(N)}(\omega)$, but again in the strong disorder regime, improving the estimates of \cite{CS09b}. 

Under some mild assumptions on the random variables and the interaction, we first prove in Theorem \ref{thm:inf.spectrum} that the bottom of the spectrum of $\BH^{(N)}(\omega)$ is non-random and equals $0$ almost surely and we emphasize that the result is both crucial to prove localization at low energy and not obvious for multi-particle systems. 

We then prove exponential and strong dynamical localization in the Hilbert-Schmidt norm near the bottom of the spectrum of $\BH^{(N)}(\omega)$ in Theorems \ref{thm:low.energy.exp.loc} and \ref{thm:low.energy.dynamical.loc} respectively. This is done under the assumption that the common probability distribution of the i.i.d. random variables is log-H\"older continuous.

The full description of $\BH^{(N)}(\omega)$ is given in Section \ref{sec:model.main.results}. In contrast to the single-particle theory, when $\BH^{(N)}(\omega)$ is restricted to disjoint cubes, the corresponding operators are not necessarily independent. This is why, following \cite{CS09b}, we introduce in Section \ref{sec:Np.scheme}, the notion of separable cubes, along with some basic geometric results, and prove a Wegner estimate, as well as an initial length scale estimate. Unfortunately, we could not directly use the multi-particle multiscale analysis developed in \cite{CS09b}, for example this scheme uses the fact that a certain parameter $p(N,g)\rightarrow\infty$ as the disorder $|g|\rightarrow \infty$. Since we are concerned with the low energy regime in this paper, we had to modify the multiscale induction, and this is the object of Section \ref{sec:MP.induction}. We finally prove our main results Theorems \ref{thm:inf.spectrum}, \ref{thm:low.energy.exp.loc} and \ref{thm:low.energy.dynamical.loc} in Section \ref{sec:proof.results}. The result on lower spectral edges Theorem \ref{thm:inf.spectrum} uses the Borel-Cantelli Lemma. The  proof of spectral localization Theorem \ref{thm:low.energy.exp.loc}, is based on the approach of von Dreifus and Klein \cite{DK89}, while the proof of dynamical localization Theorem \ref{thm:low.energy.dynamical.loc}, is based on an adaptation of the argument of Germinet and Klein \cite{GK01} to the lattice. Let us stress that the paper \cite{CS09b} did not prove dynamical localization, and the scheme we present in Section \ref{sec:low.energy.dynamical.loc} can also be used for that model.

%---------------------------------------------------------------------%
%:s.2
%---------------------------------------------------------------------%
\section{The model and the main results}\label{sec:model.main.results}

%---------------------------------------------------------------------%
%:s.2.1
%---------------------------------------------------------------------%
\subsection{Basic notations}

We are interested in the $N$-particle Hamiltonian \eqref{eq:hamiltonian} under some assumptions on $\BDelta$, $V$ and $\BU$. Fix an arbitrary integer $N\geq 2$. Since multi-scale analysis applied to this Hamiltonian requires also the consideration of Hamiltonians for $n$-particle subsystems for any $1\leq n\leq N$, we introduce notations for such $n$.
Given an integer $\nu\ge 1$, we generally equip $\DZ^{\nu}$ with the max-norm $|\,\cdot\,|$ defined by
\begin{equation}\label{eq:max.norm}
|x|=\max\{|x_1|,\dots,|x_{\nu}|\},
\end{equation}
for $x=\bigl(x_1,\dots,x_{\nu}\bigr)\in\DZ^{\nu}$. Occasionally (e.g., in the definition of the Laplacian)
we will use another norm defined by
\begin{equation}\label{eq:sum.norm}
|x|_1=|x_1|+\dots+|x_{\nu}|.
\end{equation}
Let $n\geq 1$ and $d\geq 1$ be two integers. A configuration of $n$ distinguishable quantum particles $\{x_1,\ldots,x_n\}$ in the lattice $\DZ^d$ is represented by a lattice vector $\Bx\in(\DZ^d)^n\cong \DZ^{nd}$ with coordinates $x_j=(x_j^{(1)},\dots,x_j^{(d)})\in\DZ^d$, $j=1,\dots,n$. 
 The $nd$-dimensional lattice nearest-neighbor Laplacian $\BDelta$ is defined by 
\begin{equation}           \label{eq:def.Delta}
(\BDelta\BPsi)(\Bx)
=\sum_{\substack{\By\in\DZ^{nd}\\|\By-\Bx|_1=1}}\left(\BPsi(\By)-\BPsi(\Bx)\right)
=\sum_{\substack{\By\in\DZ^{nd}\\|\By-\Bx|_1=1}}\BPsi(\By)-2dn\BPsi(\Bx),
\end{equation}
for any $\BPsi\in\ell^2(\DZ^{nd})$, $\Bx\in\DZ^{nd}$. Note that $\BDelta$ is bounded and $-\BDelta$ is nonnegative.
 We will consider  random Hamiltonians $\BH^{(n)}(\omega)$ for $n=1,\ldots,N$ of the form
\[
\BH^{(n)}(\omega)=-\BDelta+\sum_{j=1}^nV(x_j,\omega)+\BU,\tag{\ref{eq:hamiltonian}}
\]
acting in the Hilbert space $\mathcal{H}^{(n)}=\ell^2(\DZ^{dn})$. It is easy to see that \eqref{eq:max.norm} and \eqref{eq:sum.norm} are compatible with the identification $\DZ^{dn}\cong(\DZ^d)^n$, i.e., for $\{x_1,\ldots,x_n\}$ with $x_j\in\DZ^d$  the set of coordinates of $\Bx\in(\DZ^d)^n$,
\begin{align*}
|\Bx|=\max_{1\leq j\leq n}|x_j|,
%\\
\quad
|\Bx|_1=|x_1|_1+\dots+|x_n|_1.
\end{align*}

%---------------------------------------------------------------------%
%:s.2.2
%---------------------------------------------------------------------%
\subsection{Assumptions}

\subsubsection{Assumptions on $\BU$}

%-------------------%
\begin{assi*}
The potential of inter-particle interaction
\[
\BU\colon(\DZ^d)^n\to\DR
\]
is bounded and of the form
\begin{equation}\label{eq:def.U}
\BU(\Bx)=\sum_{1\leq i<j\leq n}\Phi(|x_i - x_j|),
\end{equation}
where the points $\{x_i,i=1,\ldots,n\}$ represent the coordinates of $\Bx\in(\DZ^d)^n$ and 
$\Phi\colon\DN \to \DR_+$
is a compactly supported non-negative function:
\begin{align}\label{eq:cond.U}
\exists\, r_0\in\DN:\quad
\supp\, \Phi \subset [0,r_0].
\end{align}
We will call $r_0$ the ``range'' of the interaction $\BU$.
\end{assi*}
%-------------------%

%---------------------------------------------------------------------%
\subsubsection{Assumptions on $V$}
Set $\Omega=\DR^{\DZ^d}$ and $\FB=\bigotimes_{\DZ^d}\mathcal{B}(\DR)$ where $\mathcal{B}(\DR)$ is the Borel sigma-algebra on $\DR$. Let $\mu$ be a probability measure on $\DR$ and define $\DP=\bigotimes_{\DZ^d}\mu$ on $\Omega$.

The external random potential $V\colon\DZ^d\times\Omega\to\DR$ is an i.i.d. random field
relative to  $(\Omega,\FB,\DP)$ and is defined by $V(x,\omega)=\omega_x$ for $\omega=(\omega_i)_{i\in\DZ^d}$.
The common probability distribution function, $F_V$,  of the i.i.d. random variables $V(x,\cdot)$, $x\in\DZ^d$ associated to the measure $\mu$ is defined by
\[
F_V: t \mapsto \prob{V(0,\omega)\leq t }.
\]
%-------------------%
%:(P)
%-------------------%
\begin{assp*}
It is assumed that $0\in\supp \mu\subset[0,+\infty)$. Further, the probability distribution function $F_V$ is log-H\"older continuous:  More precisely,
\begin{align}\label{eq:assumption.3prime}
&s(F_V,\varepsilon) := \sup_{a\in\DR}(F_V(a+\varepsilon)-F_V(a))
\leq\frac{C}{|\ln\epsilon|^{2A}}\\
&\text{for some }C\in(0,\infty)\text{ and }A>\frac{3}{2}\times4^Np+9Nd.\notag
\end{align}
Note that this last condition depends on the parameter $p$ which will be introduced  in Section \ref{sec:Np.scheme}.
\end{assp*}
%-------------------%

%-------------------%
%:rem 1
%-------------------%
\begin{remark}
In the course of the scale induction, the bound \eqref{eq:assumption.3prime} will be used in a situation where $\epsilon=\ee^{-L^{\beta}}$, $\beta=1/2$ and $L>0$ is a large integer. With such a value of $\epsilon$, \eqref{eq:assumption.3prime} takes the form
\begin{equation}\label{eq:assumption.3prime.L}
s(F_V,\ee^{-L^\beta})\leq C\,L^{-A}.
\end{equation}
Such a power-law estimate is sufficient for the purposes of the MSA. Further, the results of this paper are valid if the external random potential is correlated but satisfies the independence at large distance condition (IAD). In that case an additional condition will be needed in order to ensure a large deviation estimate for the random process which is automatic in the i.i.d. case \cites{K08,St01}.
\end{remark}
\subsection{Statement of the results}
For any $n=1,\ldots,N$, we denote by $\sigma(\BH^{(n)}(\omega))$ the spectrum of $\BH^{(n)}(\omega)$ and $E^{(n)}_0(\omega)$ the infimum of $\sigma(\BH^{(n)}(\omega))$.
%--------------------%
%:pro 1
%--------------------%
\begin{theorem}[The lower spectral edges are non-random]\label{thm:inf.spectrum}
Let $1\leq n\leq N$. Under  assumptions  $\condI$ and $\condP$, we have with probability one:
\[
[0,4nd]\subset\sigma(\BH^{(n)}(\omega))\subset[0,+\infty).
\]
Consequently,
\[
E^{(n)}_0:= \inf\sigma(\BH^{(n)}(\omega))=0 \quad \emph{a.s.}
\]
\end{theorem}
%--------------------%

%---------------------------------------------------------------------%
%s.2.3
%---------------------------------------------------------------------%

%-------------------%
%:thm 1
%-------------------%

%--------------------%
\begin{theorem}[Multi-particle Anderson localization at low energies]\label{thm:low.energy.exp.loc}
Under the assumptions $\condI$ and $\condP$, there exists $E^*>E_0^{(N)}$ such that with $\DP$-probability one,
\begin{enumerate}[\rm(i)]
\item
the spectrum of $\BH^{(N)}(\omega)$ in $[E_0^{(N)},E^*]$ is nonempty and pure point;
\item
any eigenfunction $\BPsi_i(\Bx,\omega)$ with eigenvalue $E_i(\omega)\in [E_0^{(N)},E^*]$ is exponentially decaying at infinity: there exist a non-random constant $m>0$ and a random constant $C_i(\omega)>0$ such that
\begin{equation} \label{eq:eigenfunctions.decay}
\left|\BPsi_i(\Bx,\omega)\right|\leq C_i(\omega)\ee^{-m|\Bx|}.
\end{equation}
\end{enumerate}
\end{theorem}
%-------------------------%

%:thm 2
%--------------------------%

%--------------------------% 
\begin{theorem}[Strong dynamical localization at low energies]\label{thm:low.energy.dynamical.loc}
Under the assumptions $\condI$ and $\condP$, there exist $E^*>E^{(N)}_0$ and $s^*(N,d)>0$ such that for any bounded $\BK\subset \DZ^{Nd}$ and any $0<s<s^*$ we have 
\begin{equation}\label{eq:low.energy.dynamical.loc}
\esm{\sup_{\|f\|_{\infty}\leq 1}\Bigl\| |\BX|^{\frac{s}{2}}f(\BH^{(N)}(\omega))\BP_{I}(\BH^{(N)}(\omega))\Bone_{\BK}\Bigr\|_{HS}^2}<\infty,
\end{equation}
where $(|\BX|\BPsi)(\Bx):=|\Bx|\BPsi(\Bx)$, $\BP_{I}(\BH^{(N)}(\omega))$ is the spectral projection of $\BH^{(N)}(\omega)$ onto the interval $I:=[E^{(N)}_0,E^*]$, and the supremum is taken over bounded measurable functions $f$.
\end{theorem}
%---------------------------%

To prove these results, especially Theorems \ref{thm:low.energy.exp.loc} and \ref{thm:low.energy.dynamical.loc} we will begin with an input: the multi-scale analysis which is the object of the following two Sections, Section \ref{sec:Np.scheme} and \ref{sec:MP.induction}.

%---------------------------------------------------------------------%
%:s.3
%---------------------------------------------------------------------%
\section{The N-particle MSA scheme }  \label{sec:Np.scheme}

%---------------------------%
%s.3.1
%---------------------------
\subsection{The general structure of the multiparticle MSA}

Following the general approach of \cites{CS09b,AW09}, we prove localization for systems with a number of particles bounded by some integer $N<\infty$. The value of $N\geq 2$ may be arbitrary, but once it is chosen, it is fixed for the rest of the scaling analysis, and several important parameters used in our scheme depend upon $N$.

We will carry out a finite induction on the number of particles $n$ varying
from $1$ (where the well-known results of the $1$-particle localization theory can be used) to $N$. In some notations we emphasize the $N$-dependence of parameters, while in other instances this dependence is omitted for simplicity.

According to the general structure of the MSA, we work with lattice \emph{rectangles}, for $\Bu\in\DZ^{nd}$ with coordinates $\{u_1,\ldots,u_n\}\subset\DZ^d$ and given $\{L_i: i=1,\ldots,n\}$,
\begin{equation}          \label{eq:cube}
\BC^{(n)}(\Bu)=\prod_{i=1}^n C^{(1)}_{L_i}(u_i),
\end{equation}
where
\[
C^{(1)}_{L_i}(u_i)=\left\{x\in\DZ^{d}:|x-u_i|\leq L_i\right\}.
\]
By $\BC^{(n)}_L(\Bu)$ we denote the $n$-particle cube, i.e.,
\[
\BC^{(n)}_L(\Bu)=\left\{\Bx\in\DZ^{nd}:|\Bx-\Bu|\leq L\right\}.
\]
 We  define the \emph{boundary} of the domain $\BC^{(n)}(\Bu)$ by
\begin{align}
\partial\BC^{(n)}(\Bu)&=\{(\Bv,\Bv')\in\DZ^{nd}\times\DZ^{nd}\,\Big|\, |\Bv-\Bv'|_1=1\text{ and }\notag\\
&\quad \text{either  $\Bv\in\BC^{(n)}(\Bu)$, $\Bv'\notin\BC^{(n)}(\Bu)$ or $\Bv\notin\BC^{(n)}(\Bu)$, $\Bv'\in\BC^{(n)}(\Bu)$}\}\label{eq:boundary},
\end{align}
its \emph{internal boundary} by 
\begin{equation}\label{eq:int.boundary}
\partial^-\BC^{(n)}(\Bu)=\left\{\Bv\in\DZ^{nd}:\dist\left(\Bv,\DZ^{nd}\setminus \BC^{(n)}(\Bu)\right)=1\right\},
\end{equation}
and its \emph{external boundary} by
\begin{equation}\label{eq:ext.boundary}
\partial^+\BC^{(n)}(\Bu)=\left\{\Bv\in\DZ^{nd}\setminus\BC^{(n)}(\Bu):\dist\left(\Bv, \BC^{(n)}(\Bu)\right)=1\right\}.
\end{equation}
Clearly, $|\BC_L^{(n)}(\Bu)| := \card\BC_L^{(n)}(\Bu)=(2L+1)^{nd}$. We will often use the simpler estimate $|\BC_L^{(n)}(\Bu)|\leq(3L)^{nd}$. Given $n\ge 1$ and $\Bu\in(\DZ^d)^n$ with coordinates $\{u_1,\ldots,u_n\}$, we define 
$$
\varPi \Bu := \{u_1, \ldots, u_n\} \subset \DZ^d,
$$
$\varPi \Bu$ will be referred to as the \emph{projection} of $\Bu$. More generally, for any nonempty subset $\CJ\subset\{1,\dots,n\}$ we define the $\CJ$-projection by
$$
\varPi_{\CJ} \Bu := \{u_j, j\in\CJ\} \subset \DZ^d.
$$
(Note that for $\CJ = \{1, \ldots, n\}$ we have $\varPi_{\CJ} = \varPi$. )
For an $n$-particle rectangle 
$$
\BC^{(n)}(\Bu)=\prod_{i=1}^n C^{(1)}_{L_i}(u_i)\subset\DZ^{nd}
$$
we can consider its $\CJ$-projection, for any $\CJ \subseteq \{1, \ldots, n\}$:
\begin{equation}   \label{eq:cube.partial.projection}
\varPi_{\CJ}\BC^{(n)}(\Bu)=\bigcup_{j\in\CJ}C^{(1)}_{L_j}(u_j)\subset\DZ^d.
\end{equation}
Again, for $\CJ = \{1, \ldots, n\}$ we will write $\varPi \BC^{(n)}(\Bu)$
instead of $\varPi_{\CJ} \BC^{(n)}(\Bu)$.
We define the restriction of the Hamiltonian $\BH^{(n)}$ to  $\BC^{(n)}(\Bu)$ by
\begin{align*}
&\BH_{\BC^{(n)}(\Bu)}^{(n)}=\BH^{(n)}\big\vert_{\BC^{(n)}(\Bu)}\\
&\text{with simple boundary conditions on }\partial^+\BC^{(n)}(\Bu),
\end{align*}
i.e., $\BH^{(n)}_{\BC^{(n)}(\Bu)}(\Bx,\By)=\BH^{(n)}(\Bx,\By)$ whenever $\Bx,\By\in\BC^{(n)}(\Bu)$ and $\BH^{(n)}_{\BC^{(n)}(\Bu)}(\Bx,\By)=0$ otherwise.
We denote the spectrum of $\BH_{\BC^{(n)}(\Bu)}^{(n)}$  by
$\sigma\bigl(\BH_{\BC^{(n)}(\Bu)}^{(n)}\bigr)$ and its resolvent by
\begin{equation}\label{eq:def.resolvent}
\BG_{\BC^{(n)}(\Bu)}(E):=\Bigl(\BH_{\BC^{(n)}(\Bu)}^{(n)}-E\Bigr)^{-1},\quad E\in\DR\setminus\sigma\Bigl(\BH_{\BC^{(n)}(\Bu)}^{(n)}\Bigr).
\end{equation}
The matrix elements $\BG_{\BC^{(n)}(\Bu)}(\Bx,\By;E)$ are usually called the
\emph{Green functions} of the operator $\BH_{\BC^{(n)}(\Bu)}^{(n)}$.
The multi-scale analysis is based on a length scale $\{L_k\}_{k\geq 0}$ which is chosen as follows.
%--------------------------%
%:def 1
%--------------------------%
\begin{definition}[length scale]\label{def:length.scale}
The length-scale $\{L_k\}_{k\geq 0}$ is a sequence of integers defined by the initial length-scale $L_0>3$, and by the recurrence relation
\[
L_{k+1}=\lfloor L_k^{\alpha}\rfloor +1,
\]
where $1<\alpha<2$ is some fixed number. In this paper, $\alpha=3/2$.
\end{definition}
%--------------------------%
This length scale $\{L_k\}_{k\geq 0}$ is assumed to be chosen at the beginning of the multi-scale
analysis, except that in the course of the analysis it is often required that $L_0$ be large enough.
%-------------------%
%:def 2
%-------------------%
\begin{definition}[$E$-resonant]
Let $n\geq 1$, $\beta=1/2$, and $E\in\DR$ be given. Consider a rectangle  $\BC^{(n)}(\Bu)=\prod_{i=1}^n C^{(1)}_{L_i}(u_i)$ and set $L=\min_{i=1,\ldots,n}\{L_i\}$. $\BC^{(n)}(\Bu)$ 
 is called $E$-resonant ($E$-R) if
\begin{equation} \label{eq:E-resonant}
\dist\Bigl[E,\sigma\bigl(\BH_{\BC^{(n)}(\Bu)}^{(n)}\bigr)\Bigr]<\ee^{-L^{\beta}}.
\end{equation}
Otherwise it is called $E$-non-resonant ($E$-NR).
\end{definition}
%-------------------%

The next definition concerns only cubes and depends on the parameter $\alpha>1$ which governs the length scale of our
multi-scale analysis.

%--------------------------%
%:def 3
%--------------------------%
\begin{definition}[$E$-completely non-resonant]\label{def:E-CNR}
Let $E\in\DR$ be given and $\alpha=3/2$. A cube $\BC_L^{(n)}(\Bv)\subset\DZ^{nd}$ of size $L\geq 2$ is called
$E$-completely non-resonant ($E$-CNR) if it does not contain any $E$-R cube of size $\geq L^{1/\alpha}$.
In particular, $\BC^{(n)}_L(\Bv)$ is itself $E$-NR.
\end{definition}
%--------------------------%

%-------------------%
%:def 4
%-------------------%
\begin{definition}[$(E,m)$-singular]
Let $m>0$ and $E\in\DR$ be given.  A cube $\BC_L^{(n)}(\Bu)\subset\DZ^{nd}$,
$1\leq n\leq N$ is called $(E,m)$-nonsingular ($(E,m)$-NS) if
\begin{equation}\label{eq:singular}
\max_{\Bv\in\partial^-\BC_L^{(n)}(\Bu)}\left|\BG_{\BC_L^{(n)}(\Bu)}(\Bu,\Bv;E)\right|
\leq\ee^{-\gamma(m,L,n)L},
\end{equation}
where $\gamma(m,L,n)=\gamma(m,L,n,N)$ is defined by
\begin{equation}\label{eq:gamma}
\gamma(m,L,n)=m(1+L^{-1/8})^{N-n+1}>m.
\end{equation}
Otherwise it is called $(E,m)$-singular ($(E,m)$-S).
\end{definition}
%-------------------%
% def 5

We will also make use of the following notion.

%-----------------------------------------%
%: def 6
%-----------------------------------------%

%------------------------------------------------%
\begin{definition}\label{def:separability}
Let  $\BC^{(n)}(\Bx)=\prod_{i=1,\ldots,n}C^{(1)}_{L_i}(x_i)$ and   $\BC^{(n)}(\By)=\prod_{i=1,\ldots,n}C^{(1)}_{L'_i}(y_i)$ be two rectangles.
\begin{enumerate}[\rm(i)]
\item
$\BC^{(n)}(\Bx)$ is $\CJ$ pre-separable from $\BC^{(n)}(\By)$ if there exists a
nonempty subset $\CJ\subset\{1,\cdots,n\}$ such that
\[
\left(\bigcup_{j\in \CJ}\varPi_{j}\BC^{(n)}(\Bx)\right)\cap
\left(\bigcup_{j\notin \CJ}\varPi_{j}\BC^{(n)}(\Bx)\cup \varPi\BC^{(n)}(\By)\right)=\emptyset.
\]
\item
A pair $(\BC^{(n)}(\Bx),\BC^{(n)}(\By))$ is pre-separable if one of the rectangle
is $\CJ$ pre-separable from the other. 
\item
A pair $(\BC^{(n)}(\Bx),\BC^{(n)}(\By))$ is separable if and only if it is pre-separable and $|\Bx-\By|>7NL$ where  $L=\max\{L_i,L'_i:i=1,\ldots,n\}$.
\end{enumerate}
\end{definition}
%-------------------------------------------------%

This notion of separability was initially used in the framework of the multiparticle MSA  in \cite{CS09b} in the
discrete case and in \cite{BCS11} in the continuum.
Before proceeding further, we state a geometric fact which describes pairs of
separable cubes.

%--------------------------------%
%: lem 1
%--------------------------------%

%--------------------------------%
\begin{lemma}\label{lem:separability}
Let $L>1$.
\begin{enumerate}[\rm(A)]
\item
For any $\Bx\in\DZ^{nd}$, there exists a collection of $n$-particle cubes
$\BC^{(n)}_{2nL}(\Bx^{(\ell)})$ with $\ell=1,\ldots,\kappa(n)$, $\kappa(n)= n^n$ such that if
 $\By$ satisfies $|\By-\Bx|>7NL$ and
\[
\By \notin \bigcup_{\ell=1}^{\kappa(n)} \BC^{(n)}_{2nL}(\Bx^{(\ell)})
\]
then the cubes $\BC^{(n)}_L(\Bx)$ and $\BC^{(n)}_L(\By)$ are separable.
\item
Let $\BC^{(n)}_L(\By)\subset \DZ^{nd}$ be an $n$-particle cube. Any cube  $\BC^{(n)}_L(\Bx)$ with $|\By-\Bx|>\max_{1\leq i,j\leq n}|y_i-y_j| +3NL$ is $\CJ$-separable from
$\BC^{(n)}_L(\By)$ for some $\CJ\subset\{1,\ldots,n\}$.
\end{enumerate}
\end{lemma}
%---------------------------------%

\begin{proof}
See the proof in the Appendix (Section \ref{sec:appendix}).
\end{proof}

Given  $1\leq n\leq N$, the following property for some $E^*>0$ will play an important role in our
strategy.

%-------------------%
\begin{dsknn*}
For any pair of separable cubes $\BC_{L_k}^{(n)}(\Bu)$ and $\BC_{L_k}^{(n)}(\Bv)$
\begin{equation}\label{eq:property.DS.k.n.N}
\DP\left\{\exists\,E\in I:\BC_{L_k}^{(n)}(\Bu),\ \BC_{L_k}^{(n)}(\Bv) \text{ are } (E,m)\text{-S}\right\}
\leq L_k^{-2p\,4^{N-n}},
\end{equation}
where $m>0$, $p>6Nd$, and $I=(-\infty,E^*]$ with $E^*>0$, are fixed.
\end{dsknn*}
%-------------------%
For brevity, we denote $\dskn:=\dsk{N}$.

By Definition \ref{def:separability}, if for $N$ fixed and  two cubes
$\BC_{L_k}^{(n)}(\Bu)$, $\BC_{L_k}^{(n)}(\Bv)$ with $1\leq n\leq N$
are separable, then $|\Bu-\Bv|> 7\,NL_k$.
The role of this lower bound on the distance between their centers will become clear later.

%-------------------%

%---------------------------------------------------------------------%
%:s.3.2
%---------------------------------------------------------------------%
\subsection{Eigenvalue concentration bounds}

%-------------------%

%-------------------%
%: thm 4
%-------------------%

%--------------------------------%
\begin{theorem}\label{thm:Wegner}
For all $\varepsilon>0$ and two pre-separable rectangles $\BC^{(n)}(\Bu)=\prod_{i=1,\ldots,n}C^{(1)}_{L_i}(u_i)$, $\BC^{(n)}(\Bu')=\prod_{i=1,\ldots,n}C^{(1)}_{L'_i}(u'_i)$, we have:
\begin{align*}
&\prob{\dist\left[\sigma(\BH^{(n)}_{\BC^{(n)}(\Bu)}),\sigma(\BH^{(n)}_{\BC^{(n)}(\Bu')})\right]
\leq \varepsilon }\\
&\qquad\leq|\BC^{(n)}(\Bu')|\cdot |\BC^{(n)}(\Bu)| \cdot\max_{i=1,\ldots,n}\cdot\max_{\Bu,\Bu'}\{ |\varPi_i\BC^{(n)}(\Bu)|,|\varPi_i\BC^{(n)}(\Bu')|\} \cdot s(F_V,2\varepsilon).
\end{align*}
\end{theorem}
%--------------------------------%

\begin{proof} See the Appendix (Section \ref{sec:appendix}).
\end{proof}
%----------------------------%
%: cor 1
%----------------------------%

%-----------------------------------%
\begin{corollary}\label{cor:Wegner}
Assume that the random potential satisfies assumption $\condP$, then for any $L\geq L_0$ and any pair of separable cubes $\BC_L^{(n)}(\Bx)$ and $\BC_L^{(n)}(\By)$,
\begin{equation}\label{eq:WS2.n}
\DP\left\{\exists\,E\in\DR,\text{ neither $\BC_L^{(n)}(\Bx)$ nor $\BC_L^{(n)}(\By)$ is $E$-CNR}\right\}<L^{-4^Np}.
\end{equation}
\end{corollary}
%-----------------------------------%

\begin{proof}
First, note that for any given pair of cubes 
$\BC^{(n)}_{\ell_1}(\Bu^1)$ and $\BC^{(n)}_{\ell_2}(\Bu^2)$ with $L^{1/\alpha}\leq \ell_1,\ell_2\leq L$,
one has
$$
\begin{aligned}
&\prob{\text{$\exists E\in\DR$: $\BC^{(n)}_{\ell_1}(\Bu^1)$ and $\BC^{(n)}_{\ell_2}(\Bu^2)$ are $E$-R}}
\\
& \quad= \prob{\text{$\exists E\in\DR:\dist[E,\sigma(\BH^{(n)}_{\BC_{\ell_j}(\Bu^j)})]<e^{-\ell_j^{\beta}}, j=1,2$}}
\\
& \quad \le \prob{
\dist[\sigma(\BH^{(n)}_{\BC_{\ell_1}(\Bu^1)}),\sigma(\BH^{(n)}_{\BC_{\ell_2}(\Bu^2)})]
   <2\ee^{-L^{\beta/\alpha}}}
\\
& \quad \leq (3L)^{(2n+1)d} \cdot \mathrm{Const} \, L^{-A/\alpha},
\end{aligned}
$$
by Theorem \ref{thm:Wegner} and assumption $\condP$. Counting the number of possible centers
$\Bu^1,\Bu^2$ in the respective cubes of radius $L$ and the number of possible
values $\ell_1,\ell_2$, we conclude that
\begin{align}
&\prob{\text{$\exists E\in \DR:$ neither $\BC^{(n)}_{L}(\Bx)$ nor $\BC^{(n)}_{L}(\By)$ is $E$-CNR}} \notag
\\
&= \DP\{\text{ $\exists E\in\DR$, $\exists \Bu^1\in\BC^{(n)}_{L}(\Bx)$, $\exists\Bu^2\in\BC^{(n)}_{L}(\By)$, $\exists \ell_1,\ell_2$ with $L^{1/\alpha}\leq \ell_1,\ell_2\leq L$}:\notag\\
&\qquad\qquad \text{$\BC^{(n)}_{\ell_1}(\Bu^1)$ and $\BC^{(n)}_{\ell_2}(\Bu^2)$ are $E$-R}\}\notag
\\
&\leq (3L)^{2nd} \cdot L^2 \cdot (3L)^{(2n+1)d}\cdot \mathrm{Const} \, L^{-A/\alpha}\notag
\le L^{-4^Np},
\end{align}
since $A>4^Np\alpha+ 9Nd$.
\end{proof}
%------------------------------------------%
%-------------------------------------------------------------------%
%:s.3.3
%-------------------------------------------------------------------%
\subsection{Initial scale estimates }\label{sec:initial.MSA.bounds}
We first recall two results

%-------------------%
%:lem 2
%-------------------%
\begin{lemma}[Combes--Thomas estimate; cf. \cite{K08}*{Theorem 11.2}]\label{lem:CT}
Consider a lattice  Schr\"o\-din\-ger operator
\[
H_\Lambda = -\Delta_\Lambda + W(x)
\]
acting in $\ell^2(\Lambda)$, $\Lambda\subset\DZ^{\nu}$, $\nu\geq 1$, with an arbitrary\footnote{This includes the cases of single- and multi-particle operators, that differ only by their potentials.} potential $W\colon\Lambda\to\DR$.
Suppose that $E\in\DR$ is such that\footnote{Theorem 11.2 from \cite{K08} is formulated with the equality  $\dist(E,\sigma(H_{\Lambda})) = \eta$, but it is clear from the proof that it remains valid if  $\dist(E,\sigma(H_{\Lambda})) \ge \eta$.} $\dist(E,\sigma(H_{\Lambda})) \ge \eta$ with $\eta \in(0,1]$.
Then 
\begin{equation}\label{eq:CT}
\forall\; x,y\in\Lambda\qquad
\left|\left(H_{\Lambda} -E\right)^{-1}(x,y)\right| \le 2 \eta^{-1}\,\ee^{-\frac{\eta}{12\nu}|x-y|}.
\end{equation}
\end{lemma}
%-------------------%

%-------------------%
%:lem 3
%-------------------%
\begin{lemma} \label{lem:low.energy.gap.prob}
Let $H^{(1)}(\omega)=-\Delta+V(x,\omega)$ be a random single-particle lattice Schr\"o\-dinger operator in $\ell^2(\DZ^d)$. Assume that the random variables $V(x,\omega)$ are i.i.d.  and nonnegative. Then, for any $C>0$ there exist  arbitrary large $L_0(C)>0$ and $C_1,c>0$ such that for any cube $C^{(1)}_{L_0}(u)$, the lowest eigenvalue $E_0^{(1)}(\omega)$ of  $H_{C^{(1)}_{L_0}(u)}^{(1)}(\omega)$  satisfies
\begin{equation}\label{eq:initial.scale}
\DP\bigl\{E_0^{(1)}(\omega)\leq 2 CL_0^{-1/2}\bigr\}\leq C_1L_0^d \ee^{-cL_0^{d/4}}.
\end{equation}
\end{lemma}
%-------------------%

%-------------------%
\begin{proof}
See equation (11.23) from the proof of Theorem 11.4 in \cite{K08}. The absolute continuity of the distribution of the random variables is not required for this result, so it applies to our model. Lemma \ref{lem:low.energy.gap.prob} follows actually from the study of Lifshitz tails and is based on a large deviation estimate valid for i.i.d. processes see \cite{K08}*{Lemma 6.4} and \cite{St01}*{Theorem 2.1.3}. 
\end{proof}
%-------------------%

In our earlier work \cite{E11}, we have already inferred the initial scale MSA estimate for the two-particle lattice Anderson model from  the results given by Lemma \ref{lem:CT} and Lemma \ref{lem:low.energy.gap.prob}. The next statement shows how it can be extended to an arbitrary number of particles.

%-------------------%
%:lem 4
%-------------------%
\begin{lemma} \label{lem:initial.scale.V.positive}
Let $\BH^{(n)}(\omega)=-\BDelta+V(x_1,\omega)+\dots+V(x_n,\omega)+\BU(\Bx)$ be an $n$-particle random Schr\"odinger operator in $\ell^2(\DZ^{nd})$ where  $\BU$ and $V$ satisfy $\condI$ and $\condP$ respectively.
Then, for any $C>0$ there exist arbitrary large $L_0(C)>0$ and $C_1,c>0$ such that for any cube $\BC^{(n)}_{L_0}(\Bu)$ the lowest eigenvalue $E_0^{(n)}(\omega)$ of   $\BH_{\BC_{L_0}^{(n)}(\Bu)}^{(n)}(\omega)$ satisfies
\begin{equation}\label{eq:initial.bound}
\DP\bigl\{E_0^{(n)}(\omega)\leq 2CL_0^{-1/2}\bigr\}\leq C_1L_0^d\ee^{-cL_0^{1/4}}.
\end{equation}
\end{lemma}
%-------------------%

%-------------------%
\begin{proof}
Since the interaction potential $\BU$ is nonnegative,
it follows from the min-max principle
that the lowest eigenvalue $E_0^{(n)}(\omega)$ of $\BH_{\BC_{L_0}^{(n)}}^{(n)}(\omega)$ is bounded from below
by the lowest eigenvalue $\widetilde E_0^{(n)}(\omega)$ of the operator
\[
\widetilde\BH^{(n)}_{\BC^{(n)}_{L_0}(\Bu)}(\omega)=-\BDelta+V(x_1,\omega)+\dots+V(x_n,\omega).
\]
Denote $\CH_i=\ell^2(C^{(1)}_{L_0}(u_i))$. The latter operator can be rewritten as follows:
\begin{equation}
\widetilde\BH^{(n)}_{\BC^{(n)}_{L_0}(\Bu)}(\omega)=\\
\sum_{j=1}^n\underbrace{\Bone_{\CH_1}\otimes\dots\otimes\Bone_{\CH_{j-1}}}_{j-1\text{ times}}\otimes H_j^{(1)}\otimes\underbrace{\Bone_{\CH_{j+1}}\otimes\dots\otimes\Bone_{\CH_n}}_{n-j\text{ times}}
\end{equation}
where $H_j^{(1)}(\omega)=-\Delta^{(j)}_{C^{(1)}_{L_0}(u_j)}+V(x_j,\omega)$, $j=1,\dots,n$ acting in $\CH_j$. Hence the lowest eigenvalue $\widetilde E_0^{(n)}(\omega)$ of $\widetilde\BH^{(n)}_{\BC^{(n)}_{L_0}(\Bu)}(\omega)$ has the following form:
\[
\widetilde E_0^{(n)}(\omega)=\sum_{j=1}^nE_{0,j}^{(1)}(\omega)
\]
where $E_{0,j}^{(1)}$ is the lowest eigenvalue of $H_j^{(1)}$. All eigenvalues $E_{0,j}^{(1)}$ are nonnegative due to the nonnegativity of the operators $H_j^{(1)}$ with nonnegative external potential, so that for any $s\geq 0$,
\[
\DP\left\{\widetilde E_0^{(n)}(\omega)\leq s\right\}\leq\DP\left\{E_{0,1}^{(1)}(\omega)\leq s\right\}.
\]
Therefore, Lemma \ref{lem:initial.scale.V.positive} follows from Lemma \ref{lem:low.energy.gap.prob} applied to the single-particle Schr\"odinger operator $H_1^{(1)}$.
\end{proof}
%-------------------%

%-------------------%
%:rem 2
%-------------------%
\begin{remark}
It is to be emphasized that for the initial scale bound, the amplitude of the interaction
potential is irrelevant,  provided that it is nonnegative. This amplitude is not used
in the scale induction, either. In other words, a nonnegative interaction enhances the
``Lifshitz tails'' phenomenon.
\end{remark}
%-------------------%

%-------------------%
%:thm 5
%-------------------%
\begin{theorem}[initial scale estimate]\label{thm:initial.estimate}
Under assumptions $\condI$ and $\condP$, for any $p>0$, there exists arbitrarily large $L_0(N,d,p)$ such that if  $m:=(14N^N+6Nd)L_0^{-1/2}$ and $E^*:=(12Nd)( 2^{N+1}m)$, then $\dsn{0,n}$ is valid for any $n=1,\ldots,N$. 
\end{theorem}
%-------------------%
%-------------------%
\begin{proof}
Set $C=12Nd\cdot2^{N+1}\left[14N^N+6Nd\right]$ and let $\BC^{(n)}_{L_0}(\Bu)$ be a cube in $\DZ^{nd}$. Consider  $\omega\in\Omega$ such that the first eigenvalue $E_0^{(n)}(\omega)$ of $\BH^{(n)}_{\BC^{(n)}_{L_0}(\Bu)}(\omega)$ satisfies $E_0^{(n)}(\omega)>2CL_0^{-1/2}$. Then for all $E\leq CL_0^{-1/2}=E^*$ we have
\[
\dist(E,\sigma(\BH^{(n)}_{\BC^{(n)}_{L_0}(\Bu)}(\omega)))=E_0^{(n)}(\omega) -E >CL_0^{-1/2}=: \eta.
\]
\noindent
For $L_0$ large enough $\eta\leq 1$ and the Combes-Thomas estimate (Lemma \ref{lem:CT}) implies that
 for any $\Bv\in\partial^-\BC^{(n)}_{L_0}(\Bu)$:
\[
|\BG_{\BC^{(n)}_{L_0}(\Bu)}(E,\Bu,\Bv)|\leq 2 \eta^{-1}\, \ee^{-\frac{CL_0^{-1/2}}{12Nd}L_0}\\.
\]
Now observe that $\frac{CL_0^{-1/2}}{12 Nd}=2^{N+1}m$. Thus
\begin{align*}
|\BG_{\BC^{(n)}_{L_0}(\Bu)}(E,\Bu,\Bv)|&\leq 2 C^{-1}L_0^{1/2}\,  \ee^{-2^{N+1} m L_0}\\
&\leq  2 C^{-1}L_0^{1/2} \ee^{-2\gamma(m,L_0,n)L_0}
\\
&\leq \ee^{-\gamma(m,L_0,n)L_0},
\end{align*}
for $L_0$ large enough, since
\[
\gamma(m,L,n)=m(1+L^{-1/8})^{N-n+1}< m\times 2^N.
\]
This implies that $\BC^{(n)}_{L_0}(\Bu)$ is $(E,m)$-NS. Thus $\omega\in\{\omega\in\Omega\,\Big|\, \BC^{(n)}_L(\Bu)$ \text{ is $(E,m)$-NS $\forall E\leq E^*$}$\}$. Therefore
\[
\prob{\text{$\exists E\leq E^*$, $\BC^{(n)}_{L_0}(\Bu)$ is $(E,m)$-S}}\leq \prob{E_0^{(n)}(\omega)\leq 2CL_0^{-1/2}},
\]
and by Lemma \ref{lem:initial.scale.V.positive},
\[
\DP\bigl\{E_0^{(n)}(\omega)\leq 2CL_0^{-1/2}\bigr\}\leq C_1L_0^d\ee^{-cL_0^{1/4}}.
\]
Finally, the quantity $C_1 L_0^d\ee^{-cL_0^d/4}$ for $L_0$ sufficiently large is less than $L_0^{-2p\,4^{N-n}}$. This proves the claim since the probability for two cubes to be singular at the same energy is bounded by the probability of either one of them to be singular.
\end{proof}
In the rest of the paper, the parameters $m$ and $E^*$ are as in Theorem \ref{thm:initial.estimate}.
%-------------------%
%---------------------------------------------------------------------%
%:s.4
%---------------------------------------------------------------------%
\section{Multi-scale induction } \label{sec:MP.induction}
We follow the general scheme described in \cite{CS09b} but we will modify it here to make it suitable with the study of the low-energy regime. In Section \ref{ssec:FI.S}, we establish some useful geometrical facts valid for any $n\geq 1$. Starting from Section \ref{ssec:PI.cubes}, we perform the induction step $n-1\leadsto n$, using scale induction in $L_k$, assuming all necessary properties of $n'$-particle systems hold with any $1\leq n'\leq n-1$ and in cubes of any size $L_j$, $j\geq 0$. Further, the induction step $L_k\leadsto L_{k+1}$ is carried out separately for three types of pairs of singular $n$-particle cubes (see Sections \ref{ssec:PI.cubes}-\ref{ssec:mixed.S}).

%---------------------------------------------------------------------%
%:s.4.1
%---------------------------------------------------------------------%
\subsection{Fully and partially interactive cubes}\label{ssec:FI.S}

%-------------------%
%:def 7
%-------------------%
\begin{definition}[fully/partially interactive]\label{def:diagonal.cubes}
An $n$-particle cube $\BC_L^{(n)}(\Bu)\subset\DZ^{nd}$ is
called fully interactive (FI) if
\begin{equation}\label{eq:def.FI}
\diam \varPi \Bu := \max_{i\ne j} |u_i - u_j| \le n(2L+r_0),
\end{equation}
and partially interactive (PI) otherwise.
\end{definition}
%-------------------%

The following simple statement clarifies the notion of PI cube.

%-------------------%
%:lem 5
%-------------------%
\begin{lemma}\label{lem:PI.cubes}
If a cube $\BC_L^{(n)}(\Bu)$ is PI, then there exists a subset $\CJ\subset\left\{1,\dots,n\right\}$ with $1\leq\card\CJ\leq n-1$ such that
\[
\dist\left(\varPi_{\CJ}\BC_L^{(n)}(\Bu),\varPi_{\CJ^{\comp}}\BC_L^{(n)}(\Bu)\right)>r_0.
\]
\end{lemma}  
%-------------------%

%-------------------%
\begin{proof}
It is convenient to use the canonical injection $\DZ^d\hookrightarrow\DR^d$; then the notion of connectedness in $\DR^d$ induces its analog for lattice cubes. Set $R:=2L+r_0$ and assume that $\diam\varPi\Bu = \max_{i,j}|u_i - u_j|>nR$.
If the union of cubes $C_{R/2}^{(1)}(u_i)$, $1\leq i\leq n$, were not decomposable into two (or more) disjoint groups, then it would be connected, hence its diameter would be bounded by $n(2(R/2))=nR$, hence $\diam \varPi\Bu\leq nR$ which contradicts the hypothesis. Therefore, there exists an index subset $\CJ\subset\{1,\ldots,n\}$ such that $|u_{j_1}-u_{j_2}|>2(R/2)$ for all $j_1\in\CJ$, $j_2\in\CJ^c$, this implies that
\begin{align*}
\dist\left(\varPi_{\CJ}\BC_L^{(n)}(\Bu),\varPi_{\CJ^{\comp}}\BC_L^{(n)}(\Bu)\right)&=\min_{j_1\in\CJ,j_2\in\CJ^c}\dist\left(C^{(1)}_L(u_{j_1}),C^{(1)}_L(u_{j_2})\right)\\
&\geq \min_{j_1\in\CJ,j_2\in\CJ^{c}}|u_{j_1}-u_{j_2}|-2L>r_0.
\end{align*}
\end{proof}
%-------------------%

If $\BC^{(n)}_L(\Bu)$ is a PI cube by the above Lemma, we can write it as 
\begin{equation}\label{eq:cartesian.cubes}
\BC_L^{(n)}(\Bu)=\BC_L^{(n')}(\Bu')\times\BC_L^{(n'')}(\Bu''),
\end{equation}
with 
\begin{equation}\label{eq:cartesian.cubes.2}
\dist\left(\varPi\BC_L^{(n')}(\Bu'),\varPi\BC_L^{(n'')}(\Bu'')\right)>r_0.
\end{equation}
where $\Bu'=\Bu_{\CJ}=(u_j:j\in\CJ)$, $\Bu''=\Bu_{\CJ^{\comp}}=(u_j:j\in\CJ^{\comp})$, $n'=\card \CJ$ and $n''=\card \CJ^{\comp}$. 

Throughout, the decomposition \eqref{eq:cartesian.cubes} will implicitly satisfy \eqref{eq:cartesian.cubes.2}.
Now we turn to geometrical properties of FI cubes.

%-------------------%
%:lem 6
%-------------------%
\begin{lemma}\label{lem:disjointness}
Let $n\geq 1$, $L>2r_0$ and consider two FI cubes $\BC_L^{(n)}(\Bu)$ and $\BC_L^{(n)}(\Bv)$ with $|\Bx-\By|>7\,nL$. Then
\begin{equation}
\varPi\BC_L^{(n)}(\Bu)\cap\varPi\BC_L^{(n)}(\Bv)=\varnothing.
\end{equation}
\end{lemma}
%-------------------%

%-------------------%
\begin{proof}
If for some $R>0$,
\[
R<|\Bx-\By|=\max_{1\leq j\leq n}|x_j-y_j|,
\]
then there exists $1\leq j_0\leq n$ such that $|x_{j_0}-y_{j_0}|>R$. Since both cubes are fully interactive, we can use \eqref{eq:def.FI} for the centers $\Bx,\ \By$ and write:
\begin{align*}
&|x_{j_0}-x_i| \le \diam \varPi \Bx \le n(2L+r_0),
\\
&|y_{j_0}-y_j| \le \diam \varPi \By \le n(2L+r_0).
\end{align*}
By triangle inequality, for any $1\leq i,j\leq n$ and $R>7\,nL >6\,nL+2\,nr_0$, we have
\begin{align*}
|x_i-y_j|
&\geq |x_{j_0}-y_{j_0}|-|x_{j_0}-x_i|-|y_{j_0}-y_j|\\
&>6nL+2nr_0-2n(2L+r_0)=2nL.
\end{align*}
Therefore, for any $1\leq i,j\leq n$,
\[
\min_{i,j}\dist(C^{(1)}_L(x_i),C^{(1)}_L(y_j))\geq \min_{i,j}|x_i-y_j|-2L>2(n-1)L\geq 0,
\]
which proves the claim.
\end{proof}
%-------------------%

Given an $n$-particle cube $\BC_{L_{k+1}}^{(n)}(\Bu)$ and $E\in\DR$, we denote
\begin{itemize}
\item
by $M_{\pai}^{\sep}(\BC_{L_{k+1}}^{(n)}(\Bu),E)$ the maximal number of pairwise separable, 
$(E,m)$-singular PI cubes $\BC_{L_k}^{(n)}(\Bu^{(j)})\subset\BC_{L_{k+1}}^{(n)}(\Bu)$;
\item
by  $M_{\pai}(\BC_{L_{k+1}}^{(n)}(\Bu),E)$ the maximal number of (not necessarily separable)
$(E,m)$-singular PI cubes $\BC^{(n)}_{L_k}(\Bu^{(j)})\subset \BC^{(n)}_{L_{k+1}}(\Bu)$ with $|\Bu^{(j)}-\Bu^{(j')}|>7NL_k$ for all $j\neq j'$;
\item
by $M_{\fui}(\BC_{L_{k+1}}^{(n)}(\Bu),E)$ the maximal number of 
$(E,m)$-singular FI cubes  $\BC_{L_k}^{(n)}(\Bu^{(j)})\subset \BC_{L_{k+1}}^{(n)}(\Bu)$ with $|\Bu^{(j)}-\Bu^{(j')}|>7NL_k$ for all $j\neq j'$\footnote{Note that by lemma \ref{lem:disjointness}, two FI cubes $\BC^{(n)}_{L_k}(\Bu^{(j)})$ and $\BC^{(n)}_{L_k}(\Bu^{(j')})$ with $|\Bu^{(j)}-\Bu^{(j')}|>7NL_k $ are automatically separable},
\item
$M_{\pai}(\BC_{L_{k+1}}^{(n)}(\Bu),I):=\sup_{E\in I}M_{\pai}(\BC_{L_{k+1}}^{(n)}(\Bu),E)$.
\item
$M_{\fui}(\BC_{L_{k+1}}^{(n)}(\Bu),I):=\sup_{E\in I}M_{\fui}(\BC_{L_{k+1}}^{(n)}(\Bu),E)$.
\item
by $M(\BC_{L_{k+1}}^{(n)}(\Bu),E)$ the maximal number of 
$(E,m)$-singular cubes  $\BC_{L_k}^{(n)}(\Bu^{(j)})\subset \BC_{L_{k+1}}^{(n)}(\Bu)$ with $|\Bu^{(j)}-\Bu^{(j')}|>7NL_k$ for all $j\neq j'$.
\item
by $M^{\sep}(\BC_{L_{k+1}}^{(n)}(\Bu),E)$ the maximal number of pairwise separable
$(E,m)$-singular cubes  $\BC_{L_k}^{(n)}(\Bu^{(j)})\subset \BC_{L_{k+1}}^{(n)}(\Bu)$ 
\end{itemize}
Clearly
\[
M_{\pai}(\BC_{L_{k+1}}^{(n)}(\Bu),E)+M_{\fui}(\BC_{L_{k+1}}^{(n)}(\Bu),E)\geq M(\BC_{L_{k+1}}^{(n)}(\Bu),E).
\]
%------------------------------%
%: lem 7
%------------------------------%

%-------------------------------------------%
\begin{lemma}\label{lem:MPI}
If $M(\BC^{(n)}_{L_{k+1}}(\Bu),E)\geq \kappa(n)+2$ with $\kappa(n)=n^n$ introduced in Lemma \ref{lem:separability}, then $M^{\sep}(\BC^{(n)}_{L_{k+1}}(\Bu),E)\geq 2$.\\
\noindent Similarly, if $M_{\pai}(\BC^{(n)}_{L_{k+1}}(\Bu),E)\geq \kappa(n)+2$  then $M_{\pai}^{\sep}(\BC^{(n)}_{L_{k+1}}(\Bu),E)\geq 2$.
\end{lemma}
%-------------------------------------------%

%-------------------------------------------%
\begin{proof}
Assume that $M^{\sep}(\BC^{(n)}_{L_{k+1}}(\Bu),E)<2$, (i.e., there is no pair of separable  cubes of radius $L_k$ in $\BC^{(n)}_{L_{k+1}}(\Bu)$), but $M(\BC^{(n)}_{L_{k+1}}(\Bu),E)\geq \kappa(n)+2$. Then $\BC^{(n)}_{L_{k+1}}(\Bu)$ must contain at least $\kappa(n)+2$ cubes $\BC^{(n)}_{L_k}(\Bv_i)$, $0\leq i\leq \kappa(n)+1$ which are non separable but satisfy $|\Bv_i-\Bv_{i'}|>7NL_k$, for all $i\neq i'$. On the other hand, by Lemma \ref{lem:separability} there are at most $\kappa(n)$ cubes $\BC^{(n)}_{2nL_k}(\By_i)$, such that any cube $\BC^{(n)}_{L_k}(\Bx)$ with $\Bx\notin \bigcup_{j} \BC^{(n)}_{2nL_k}(\By_j)$ is separable from $\BC^{(n)}_{L_k}(\Bv_0)$. Hence $\Bv_i\in \bigcup_{j} \BC^{(n)}_{2nL_k}(\By_j)$ for all $i=1,\ldots,\kappa(n)+1$.  But since for all $i\neq i'$, $|\Bv_i-\Bv_{i'}|>7NL_k$, there must be at most one center $\Bv_i$ per cube $\BC^{(n)}_{2nL_k}(\By_j)$, $1\leq j\leq \kappa(n)$, hence we come to a contradiction:
\[
\kappa(n)+1\leq \kappa(n)
\]
indeed $\Bv_i\in\bigcup_j\BC^{(n)}_{2nL_k}(\By_j)$, $i=1,\cdots,\kappa(n)+1$. The same analysis holds true if we consider only PI cubes.
\end{proof}
%--------------------%
\begin{definition}
Let $1\leq n\leq N$. We say that two subsets $\BA,\BB\subset\DZ^{nd}$ \emph{touch} if
$\BA\cap\BB\neq\emptyset$ or if there exist $\Bx\in\BA$ and $\By\in\BB$ such that $|\Bx-\By|=1$.
\end{definition}
In the following, we adapt Lemma 4.2 from \cite{DK89} to our multi-particle case. Note that in the following Lemma the term $M(\BC^{(n)}_{L_{k+1}}(\Bu),E)$ depends on the parameter $m$ which equals $(14N^N+6Nd)L_0^{-1/2}$ from Theorem \ref{thm:initial.estimate}. 

%-------------------%
%:lem 8
%-------------------%
\begin{lemma}\label{lem:CNR.NS}
Let $J=\kappa(n)+5$ with $\kappa(n)=n^n$ and $E\in\DR$. Suppose that
\begin{enumerate}[\rm(i)]
\item
$\BC_{L_{k+1}}^{(n)}(\Bu)$ is $E$-CNR,
\item
$M(\BC_{L_{k+1}}^{(n)}(\Bu),E)\leq J$.
\end{enumerate}
Then there exists $\tilde{L}_2^*(J,N,d)>0$ such that if $L_0\geq \tilde{L}_{2}^*(J,N,d)$ we have that $\BC^{(n)}_{L_{k+1}}(\Bu)$ is $(E,m)$-NS.
\end{lemma}
%-------------------%

\begin{proof}
By (ii) there are at most $J$  cubes of size $L_k$ contained in $\BC^{(n)}_{L_{k+1}}(\Bu)$ and with centers at distance $>7NL_k$ that are $(E,m)$-S. Therefore we can find $\Bx_i\in\BC^{(n)}_{L_{k+1}}(\Bu)$ with\\ $\dist(\Bx_i,\partial^-\BC^{(n)}_{L_{k+1}}(\Bu))\geq L_k$, $i=1,\cdots,r\leq J$ such that if
\[
\Bx\in\BC^{(n)}_{L_{k+1}}(\Bu)\setminus \bigcup_{i=1}^r\BC^{(n)}_{7NL_k}(\Bx_i),
\]
then $\BC^{(n)}_{L_{k}}(\Bx)$ is $(E,m)$-NS. 

Further, we claim that there exist cubes $\BC^{(n)}_{\ell_i}\subset \BC^{(n)}_{L_{k+1}}(\Bu)$ with side $\ell_i\in\{7NL_k+j(14NL_{k}+1): j=0,\cdots,J-1\}$, $i=1,\cdots,t\leq r$, such that   $\BC^{(n)}_{\ell_i}$, $\BC^{(n)}_{\ell_j}$ do not touch if $i\neq j$,
\[
\bigcup_{i=1}^r\BC^{(n)}_{7NL_k}(\Bx_i)\subset\bigcup_{i=1}^t\BC^{(n)}_{\ell_i}\quad \text{and}\quad \sum_{j=1}^t\ell_j\leq 7NL_k+J(14NL_k+1).
\]
Indeed if the $\BC^{(n)}_{7NL_k}(\Bx_i)$ do not touch then we are done. Otherwise there exist $j,s\in\{1,\cdots,r\}$ such that $\BC^{(n)}_{7NL_k}(\Bx_j)$ and $\BC^{(n)}_{7NL_k}(\Bx_{s})$ touch. Therefore, we may construct a big cube $\BC^{(n)}_{\ell_j}$ centered at $\Bx_j$ of minimal radius and containing $\BC^{(n)}_{7NL_k}(\Bx_{s})$. In fact, we set $\BC^{(n)}_{\ell_j}=\BC^{(n)}_{21NL_k+1}(\Bx_j)$, since the cubes touch, there exist points $\By\in\BC^{(n)}_{7NL_k}(\Bx_j)$ and $\By'\in\BC^{(n)}_{7NL_k}(\Bx_{s})$ such that $|\By-\By'|\leq 1$, then for $\Bz\in\BC^{(n)}_{7NL_k}(\Bx_{s})$ we have
\begin{align*}
|\Bz-\Bx_j|&\leq |\Bz-\Bx_{s}|+|\Bx_{s}-\By'|+|\By'-\By|+|\By-\Bx_j|\\
&\leq  7NL_k + (14NL_k+1)=21NL_k +1
\end{align*}
If $\BC^{(n)}_{\ell_j}$ and $\BC^{(n)}_{7NL_k}(\Bx_i)$ $i\neq j,s$ do not touch then we are 
done, otherwise we apply again the above construction reducing the number of the remaining cubes. 
Repeating this procedure we finally get the assertion. 

It follows that 
if $\Bx\in\BC^{(n)}_{L_{k+1}}(\Bu)\setminus \bigcup_{i=1}^t\BC^{(n)}_{\ell_i}$,
$\dist(\Bx,\partial^-\BC^{(n)}_{L_{k+1}}(\Bu))\geq L_k$, we have that $\BC^{(n)}_{L_k}(\Bx)$ is 
$(E,m)$-NS. Also notice that, if $\Bx\in\partial^+\BC^{(n)}_{\ell_j}$ for some $j=1,\cdots,t$ 
then $\Bx\notin\bigcup_{i=1}^t\BC^{(n)}_{\ell_i}$. 

Next, let $\Lambda\subset \BC^{(n)}_{L_{k+1}}(\Bu)$ 
and  $\Bx\in\Lambda$, $\By\in\BC^{(n)}_{L_{k+1}}(\Bu)\setminus \Lambda$, it follows from 
the resolvent identity that
\[
\BG^{(n)}_{\BC^{(n)}_{L_{k+1}}(\Bu)}(\Bx,\By;E)
= \sum_{(\Bz,\Bz')\in\partial \Lambda}\BG^{(n)}_{\Lambda}(\Bx,\Bz;E)\BG^{(n)}_{\BC^{(n)}_{L_{k+1}}(\Bu)}(\Bz',\By;E).
\]
Thus
\begin{equation}\label{eq:GRI.1}
|\BG^{(n)}_{\BC^{(n)}_{L_{k+1}}(\Bu)}(\Bx,\By;E)|\leq \left[\sum_{\Bz\in\partial^-\Lambda}|\BG^{(n)}_{\Lambda}(\Bx,\Bz;E)|\right]|\BG^{(n)}_{\BC^{(n)}_{L_{k+1}}}(\Bz_1,\By;E)|
\end{equation}
for some $z_1\in \partial^+\Lambda$.

Fix $\Bv\in\partial^-\BC^{(n)}_{L_{k+1}}(\Bu)$ and let $\Bx\in\BC^{(n)}_{L_{k+1}}(\Bu)$, with
$\dist(\Bx,\partial^-\BC^{(n)}_{L_{k+1}}(\Bu))\geq L_k$. We have two cases:
\begin{enumerate}[\rm(a)]
\item
$\BC^{(n)}_{L_k}(\Bx)$ is $(E,m)$-NS. In this case
\begin{equation}\label{eq:GRI.2}
\sum_{\Bz\in\partial^-\BC^{(n)}_{L_k}(\Bx)}|\BG^{(n)}(\Bx,\Bz;E)|\leq 2^{nd}\cdot nd \cdot L_k^{nd-1}\ee^{-\gamma(m,L_k,n)L_k},
\end{equation}
and by \eqref{eq:GRI.1}
\begin{equation}\label{eq:GRI.3}
|\BG^{(n)}_{\BC^{(n)}_{L_{k+1}}(\Bu)}(\Bx,\Bv;E)|\leq 2^{nd}\cdot nd \cdot L_k^{nd-1}\ee^{-\gamma(m,L_k,n)L_k} |\BG^{(n)}_{\BC^{(n)}_{L_{k+1}}}(\Bz_1,\By;E)|,
\end{equation}
for some $z_1\in\partial^+\BC^{(n)}_{L_k}(\Bx)$.
\item
$\BC^{(n)}_{L_k}(\Bx)$ is $(E,m)$-S. In this case we have that $\Bx\in\BC^{(n)}_{\ell_i}$ for some $i=1,\cdots,t$. If $\dist(\BC^{(n)}_{\ell_i},\partial^-\BC^{(n)}_{L_{k+1}}(\Bu))\geq L_k+1$, then equation \eqref{eq:GRI.1} gives
\[
|\BG^{(n)}_{\BC^{(n)}_{L_{k+1}}(\Bu)}(\Bx,\Bv;E)|\leq \left[ \sum_{\Bz\in\partial^-\BC^{(n)}_{\ell_i}}|\BG^{(n)}_{\BC^{(n)}_{\ell_i}}(\Bx,\Bz;E)|\right]|\BG^{(n)}_{\BC^{(n)}_{L_{k+1}}(\Bu)}(\Bz_1,\Bv;E)|
\]
for some  $\Bz_1\in\partial^+\BC^{(n)}_{\ell_i}$. Observe that $\dist(\BC^{(n)}_{\ell_i},\partial^-\BC^{(n)}_{L_{k+1}}(\Bu))\geq L_k+1$, implies  
\[
\dist(\Bz_1,\partial^-\BC^{(n)}_{L_{k+1}}(\Bu))\geq L_k.
\]
Since the cubes $\BC^{(n)}_{\ell_i}$ are $E$-NR we have:
\begin{align*}
|\BG^{(n)}_{\BC^{(n)}_{L_{k+1}}(\Bu)}(\Bx,\Bv;E)|&\leq 2^{Nd}\cdot Nd(7NL_k+J(14NL_k+1))^{nd-1}\\
&\times\ee^{(7NL_k+J(14NL_k+1))^{1/2}}|\BG^{(n)}_{\BC^{(n)}_{L_{k+1}}(\Bu)}(\Bz_1,\Bv;E)|
\end{align*}
and since the cube $\BC^{(n)}_{L_k}(\Bz_1)$ is $(E,m)$-NS, we apply  \eqref{eq:GRI.1} and \eqref{eq:GRI.2} getting
\begin{align*}
|\BG^{(n)}_{\BC^{(n)}_{L_{k+1}}(\Bu)}(\Bx,\Bv;E)|&\leq (2^{Nd}\cdot Nd)^2(7NL_k+J(14NL_k+1))^{nd-1}L_k^{nd-1}\\
&\times\ee^{(7NL_k+J(14NL_k+1))^{1/2}-\gamma(m,L_k,n)L_k}|\BG^{(n)}_{\BC^{(n)}_{L_{k+1}}(\Bu)}(\Bz_2,\Bv;E)|\\
&\leq (2^{Nd}\cdot Nd)^2 (21J)^{nd-1}(L_k+1)^{2(nd-1)}\\
&\times \ee^{(2J(14NL_k+1))^{1/2}-\gamma(m,L_k,n)L_k} |\BG^{(n)}_{\BC^{(n)}_{L{k+1}}(\Bu)}(\Bz_2,\Bv;E)|
\end{align*}
for some $\Bz_2\in\partial^+\BC^{(n)}_{L_k}(\Bz_1)$. Thus
\begin{equation}\label{eq:GRI.4}
|\BG^{(n)}_{\BC^{(n)}_{L_{k+1}}(\Bu)}(\Bx,\Bv;E)|\leq \ee^{-\gamma'(m,L_k,n)L_k}|\BG^{(n)}_{\BC^{(n)}_{L_{k+1}}(\Bu)}(\Bz_2,\Bv;E)|
\end{equation}
where
\begin{align*}
\gamma'(m,L_k,n)&=\gamma(m,L_k,n)-\frac{1}{L_k}[(2J(14NL_k+1))^{1/2} +2(nd-1)\ln(L_k+1)\\
&+\ln((2^{Nd}\cdot Nd)^2)(21J)^{nd-1})]
\\
&\geq \gamma(m,L_k,n)-\left[(30JN)^{1/2}+2Nd +2\ln(2^{Nd}Nd)\right]L_k^{-1/2}
\\
&>m-\left[(30(n^n +5)N)^{1/2}+2Nd+2\ln(2^{Nd}Nd)\right]L_0^{-1/2}> 0,
\end{align*}
since $m=(14N^N+6Nd)L_0^{-1/2}$.
\end{enumerate}

If $\Bx\in\BC^{(n)}_{L_{k+1}}(\Bu)$, $\dist\left(\Bx,\partial^-\BC^{(n)}_{L_{k+1}}(\Bu)\right)\geq L_k$, let 
\[
W(\Bx)=\left\{
\begin{array}{ll}
\cr2^{nd}\cdot nd\cdot L_k^{nd-1}\cdot\ee^{-\gamma(m,L_k,n)L_k} \text{ if $\Bx$ is as case (a)}\\
\cr\ee^{-\gamma'(m,L_k,n)L_k} \text{ if $\Bx$ is as case (b)}.
\end{array}
\right.
\]
Then \eqref{eq:GRI.3} for case (a) and \eqref{eq:GRI.4} for case (b) say that 
\[
|\BG^{(n)}_{\BC^{(n)}_{L_{k+1}}(\Bu)}(\Bx,\Bv;E)|\leq W(\Bx)|\BG^{(n)}_{\BC^{(n)}_{L_{k+1}}(\Bu)}(\Bz,\Bv;E)|,
\]
for some $\Bz\in\BC^{(n)}_{L_{k+1}}(\Bu)$.

To estimate $|\BG^{(n)}_{\BC^{(n)}_{L_{k+1}}(\Bu)}(\Bu,\Bv;E)|$ for $\Bv\in\partial^-\BC^{(n)}_{L_{k+1}}(\Bu)$, we start from $\Bu$  and find $\Bz_1,\Bz_2,\ldots$ by applying the  above procedure repeatedly when possible, getting after $r$ steps,
\begin{gather*}
|\BG^{(n)}_{\BC^{(n)}_{L_{k+1}}(\Bu)}(\Bu,\Bv;E)|\leq W(\Bu)|\BG^{(n)}_{\BC^{(n)}_{L_{k+1}}(\Bu)}(\Bz_1,\Bv;E)|\leq W(\Bu)W(\Bz_1)|\BG^{(n)}_{\BC^{(n)}_{L_{k+1}}(\Bu)}(\Bz_2,\Bv;E)|\\
\leq\cdots\leq  W(\Bu)W(\Bz_1)\cdots W(\Bz_{r-1})|\BG^{(n)}_{\BC^{(n)}_{L_{k+1}}(\Bu)}(\Bz_r,\Bv;E)|.
\end{gather*}

For this to be possible, $\Bz_1,\ldots \Bz_{r-1}$ must satisfy the conditions of either (a) or (b). Observe that in the first step, i.e, for $\Bx=\Bu$, either case (a) or (b) is done. Indeed, if $\BC^{(n)}_{L_k}(\Bu)$ is $(E,m)$-NS then we are in case (a). Otherwise $\BC^{(n)}_{L_k}(\Bu)$ is $(E,m)$-S. So $\Bu\in\BC^{(n)}_{\ell_i}$ for some $i$. Next, denote by $\diam(\BC^{(n)}_{\ell_i})$ the diameter of $\BC^{(n)}_{\ell_i})$. Then
\begin{align*}
\dist(\BC^{(n)}_{\ell_i},\partial^-\BC^{(n)}_{L_{k+1}}(\Bu))&\geq L_{k+1}-\diam(\BC^{(n)}_{\ell_i})\\
&\geq L_{k+1}-2(7NL_k+J(14NL_k+1))>L_{k+1}-44JNL_k>L_k+1,
\end{align*}
if $L_0>C(N,J)$ for some constant $C(N,J)>0$. Thus the conditions of case (b) are satisfied.

Let $r_1$ and $r_2$  be the number of times we had the bounds \eqref{eq:GRI.3} and \eqref{eq:GRI.4}, respectively. We have
\begin{align*}
|\BG^{(n)}_{\BC^{(n)}_{L_{k+1}}(\Bu)}(\Bu,\Bv;E)|&\leq 2^{Nd}\cdot Nd\cdot L_k^{nd-1}(\ee^{-\gamma(m,L_k,n)L_k})^{r_1}(\ee^{-\gamma'(m,L_k,n)L_k})^{r_2}\\
&\quad \times |\BG^{(n)}_{\BC^{(n)}_{L_{k+1}}(\Bu)}(\Bz_{r_1+r_2+1},\Bv;E)|.
\end{align*}

Notice that we have a lower bound for $r_1$ corresponding to the case where all the bad cubes $\BC^{(n)}_{\ell_i}$ are met, namely:
\begin{align*}
r_1\geq \frac{L_{k+1}-\sum_i (2\ell_i)}{L_k}&\geq \frac{L_{k+1}-2(7NL_k+J(14NL_k+1))}{L_k}
\\
&\geq L_{k+1}L_k^{-1} -44 JN.
\end{align*}
By the assumption (i) of the Lemma, $\BC^{(n)}_{L_{k+1}}(\Bu)$ is $E$-NR, so $\|\BG_{\BC^{(n)}_{L_{k+1}}(\Bu)}(E)\|<\ee^{L_{k+1}^{1/2}}$. Since  $\gamma'(m,L_k,n)>0$,  $\ee^{-r_2\gamma'(m,L_k,n)L_k}\leq 1$. Therefore 
\begin{align*}
|\BG^{(n)}_{\BC^{(n)}_{L_{k+1}}(\Bu)}(\Bu,\Bv;E)|&\leq \left( 2^{Nd}\cdot Nd\cdot L_k^{nd-1}\ee^{-\gamma(m,L_k,n)L_k}\right)^{r_1}\ee^{L_{k+1}^{1/2}}\\
&\leq \ee^{-m'L_{k+1}}
\end{align*}
where
\[
m'= \frac{1}{L_{k+1}}\left(r_1\gamma(m,L_k,n)L_k-r_1\ln(2^{Nd}Nd)L_k^{nd-1}\right)-\frac{1}{L_{k+1}^{1/2}}.
\]
Since 
$L_{k+1}L_k^{-1}- 44JN \le r_1\le L_{k+1}L_k^{-1}$
we obtain
\begin{align*}
m'&\geq \gamma(m,L_k,n)-\gamma(m,L_k,n)\, \frac{44JNL_k}{L_{k+1}}
\\
\qquad &-\frac{1}{L_{k+1}}\frac{L_{k+1}}{L_k}\ln(2^{Nd}Nd)L_k^{nd-1})-\frac{1}{L_{k+1}^{1/2}}
\\
&\geq \gamma(m,L_k,n)-\gamma(m,L_k,n)\, 44JN L_k^{-1/2}
\\
&\qquad -L_k^{-1}(\ln(2^{Nd}Nd))+(nd-1)\ln(L_k))-L_{k}^{-3/4}
\\
&\geq \gamma(m,L_k,n)[1-(44JN+\ln(2^{Nd}Nd)+Nd)L_k^{-1/2}]
\end{align*}

if $L_0\geq L_2^*(J,N,d)$ for some $L^*_2(J,N,d)>0$ large enough. Since $\gamma(m,L_k,n)=m(1+L_k^{-1/8})^{N-n+1}$,
\[
\frac{\gamma(m,L_k,n)}{\gamma(m,L_{k+1},n)}=\left(\frac{1+L_k^{-1/8}}{1+L_k^{-3/16}}\right)^{N-n+1}\geq \frac{1+L_k^{-1/8}}{1+L_k^{-3/16}}
\]
Therefore we can compute
\begin{align*}
&\frac{\gamma(m,L_k,n)}{\gamma(m,L_{k+1},n)}(1-(44JN+\ln(2^{Nd}Nd)+Nd)L_k^{-1/2})\\
 &\qquad \geq\frac{1+L_k^{-1/8}}{1+L_k^{-3/16}}(1-(44JN+\ln(2^{Nd}Nd)+Nd)L_k^{-1/2})>1,
\end{align*}
provided $L_0\geq \tilde{L}_2^*(J,N,d)$ for some large enough $\tilde{L}_2^*(J,N,d)>0$. Finally, we obtain that $m'>\gamma(m,l_{k+1},n)$ and $|\BG_{\BC^{(n)}_{L_{k+1}}(\Bu)}(\Bu,\Bv;E)|\leq \ee^{-\gamma(m,L_{k+1},n)L_{k+1}}$. This completes the proof of Lemma \ref{lem:CNR.NS}.
\end{proof}
%----------------------%
%--------------------------------------------%
Given $1\leq n\leq N$, assuming $\dsn{k-1,n'}$ for all $n'<n$ and $\dsn{k,n'}$ for any $1\leq n'\leq n$, we will prove $\dskonn$, separately for the following three types of pairs of cubes:

\begin{enumerate}[(I)]
\item
$\BC_{L_{k+1}}^{(n)}(\Bx)$ and $\BC_{L_{k+1}}^{(n)}(\By)$ are both PI-cubes.
\item
$\BC_{L_{k+1}}^{(n)}(\Bx)$ and $\BC_{L_{k+1}}^{(n)}(\By)$ are both FI-cubes.
\item
One of the cubes $\BC_{L_{k+1}}^{(n)}(\Bx)$ or $\BC_{L_{k+1}}^{(n)}(\By)$ is PI, while the other is FI.
\end{enumerate}
In the rest of this section we will denote by $I$ the interval $(-\infty,E^*]$.
%---------------------------------------------------------------------%
%:s.4.2
%---------------------------------------------------------------------%
\subsection{Pairs of partially interactive cubes}\label{ssec:PI.cubes}
Let  $\BC_{L_{k+1}}^{(n)}(\Bu)=\BC^{(n')}_{L_{k+1}}(\Bu')\times\BC^{(n'')}_{L_{k+1}}(\Bu'')$ be a PI-cube. We also write $\Bx=(\Bx',\Bx'')$ for any point $\Bx\in\BC_{L_{k+1}}^{(n)}(\Bu)$, in the same way as $(\Bu',\Bu'')$. So  the corresponding Hamiltonian $\BH^{(n)}_{\BC_{L_{k+1}}^{(n)}(\Bu)}$ is written in the form:
\begin{equation}\label{eq:decomp,H}
\BH_{\BC_{L_{k+1}}^{(n)}(\Bu)}^{(n)}\BPsi(\Bx)=(-\BDelta\BPsi)(\Bx)+\left[\BU(\Bx')+\BV(\Bx',\omega)+\BU(\Bx'')+\BV(\Bx'',\omega)\right]\BPsi(\Bx)
\end{equation}
or, in compact form
\[
\BH_{\BC_{L_{k+1}}^{(n)}(\Bu)}^{(n)}=\BH_{\BC_{L_{k+1}}^{(n')}(\Bu')}^{(n')}\otimes\mathbf{I}+ \mathbf{I}\otimes \BH_{\BC_{L_{k+1}}^{(n'')}(\Bu'')}^{(n'')}.
\]
We denote by $\BG^{(n')}(\Bu',\Bv';E)$ and $\BG^{(n'')}(\Bu'',\Bv'';E)$ the corresponding Green functions, respectively. Introduce the following notions
\begin{definition}[\cite{KN13}]\label{def:HNR}
Let $1\leq n\leq N$ and $E\in\DR$. Consider a PI cube  $\BC^{(n)}_L(\Bu)=\BC^{(n')}_L(\Bu')\times\BC^{(n'')}_L(\Bu'')$. Then $\BC^{(n)}_L(\Bu)$ is called $E$-highly non resonant ($E$-HNR) if
\begin{enumerate}[\rm(i)]
\item 
for all $\mu_j\in\sigma(\BH^{(n'')}_{\BC^{(n'')}_L(\Bu'')})$, the cube  $\BC^{(n')}_L(\Bu')$ is $(E-\mu_j)$-CNR and 
\item
for all $\lambda_i\in\sigma(\BH^{(n')}_{\BH^{(n')}_L(\Bu')})$ the cube $\BC^{(n'')}_L(\Bu'')$ is $(E-\lambda_i)$-CNR.
\end{enumerate}
\end{definition}

%-------------------%
%:def 8
%-------------------%
\begin{definition}[$(E,m)$-tunnelling]    \label{def:tunnelling}
Let $1\leq n\leq N$,  $E\in\DR$ and $m>0$. Consider a PI cube $\BC^{(n)}_L(\Bu)=\BC^{(n')}_L(\Bu')\times \BC^{(n'')}_L(\Bu'')$.

Then $\BC^{(n)}_L(\Bu)$ is called
\begin{enumerate}[(i)]
\item
$(E,m)$ left-tunnelling ($(E,m)$-LT) if
$\exists \mu_j\in\sigma(\BH^{(n'')}_{\BC^{(n'')}_L(\Bu'')})$ such that $\BC^{(n')}_L(\Bu')$ contains two separable  $(E-\mu_j,m)$-S cubes $\BC^{(n')}_{l}(\Bv_1)$ and $\BC^{(n')}_{l}(\Bv_2)$ with $L=\lfloor l^{\alpha}\rfloor+1$. Otherwise, it is called $(E,m)$ non-left-tunnelling ($(E,m)$-NLT).
\item
$(E,m)$ right-tunnelling ($(E,m)$-RT) if  $\exists \lambda_i\in\sigma(\BH^{(n')}_{\BC^{(n')}_L(\Bu')})$ such that $\BC^{(n'')}_L(\Bu'')$ contains two separable $(E-\lambda_i,m)$-S cubes  $\BC^{(n'')}_{l}(\Bv_1)$ and $\BC^{(n'')}_{l}(\Bv_2)$ with $L=\lfloor l^{\alpha}\rfloor+1$. Otherwise, it is called $(E,m)$ non-right-tunnelling ($(E,m)$-NRT).
\item
$(E,m)$-tunnelling ($(E,m)$-T) if either it is $(E,m)$-LT or $(E,m)$-RT.
Otherwise it is called $(E,m)$-non-tunnelling ($(E,m)$-NT).
\end{enumerate}
\end{definition}

%-------------------%
We reformulate and prove Lemma 3.18 from \cite{KN13} in our context.

\begin{lemma}\label{lem:HNR}
Let $E\in\DR$. If a PI cube $\BC^{(n)}_L(\Bu)=\BC^{(n')}_L(\Bu')\times\BC^{(n'')}_L(\Bu'')$  is not $E$-HNR then
\begin{enumerate}[\rm(i)]
\item
either there exist $L^{1/\alpha}\leq \ell\leq L$, $\Bx\in\BC^{(n')}_L(\Bu')$ such that  the $n$-particle rectangle $\BC^{(n)}=\BC^{(n')}_{\ell}(\Bx)\times\BC^{(n'')}_L(\Bu'')\subset\BC^{(n)}_L(\Bu)$ is $E$-R.
\item
or there exist $L^{1/\alpha}\leq \ell\leq L$, $\Bx\in\BC^{(n'')}_L(\Bu'')$ such that  the $n$-particle rectangle $\BC^{(n)}=\BC^{(n')}_{L}(\Bu')\times\BC^{(n'')}_{\ell}(\Bx)\subset\BC^{(n)}_L(\Bu)$ is $E$-R.
\end{enumerate}
\end{lemma}
\begin{proof}
By Definition \ref{def:HNR}, if $\BC^{(n)}_L(\Bu)$ is not $E$-HNR then either (a) there exists $\mu_j\in\sigma(\BH^{(n'')}_{\BC^{(n'')}_L(\Bu'')})$ such that  $\BC^{(n')}_L(\Bu')$ is not $E-\mu_j$-CNR or (b) there exists $\lambda_i\in\sigma(\BH^{(n')}_{\BC^{(n')}_L(\Bu')})$ such that $\BC^{(n'')}_L(\Bu'')$ is not $E-\lambda_i$-CNR. Let us first focus on case (a). Since $\BC^{(n')}_L(\Bu')$ is not $E-\mu_j$-CNR there exists $L^{1/\alpha}\leq \ell\leq L$, $\Bx\in \BC^{(n')}_L(\Bu')$ such that $\BC^{(n')}_{\ell}(\Bx)\subset\BC^{(n')}_L(\Bu')$ and $\BC^{(n')}_{\ell}(\Bx)$ is $E-\mu_j$-R. So $\dist(E-\mu_j,\sigma(\BH^{(n')}_{\BC^{(n')}_{\ell}(\Bx)}))<\ee^{-\ell^{\beta}}$. Therefore there exists $\eta \in\sigma(\BH^{(n')}_{\BC^{(n')}_{\ell}(\Bx)})$ such that $|E-\mu_j-\eta|<\ee^{-\ell^{\beta}}$. Now consider $\BC^{(n)}=\BC^{(n')}_{\ell}(\Bx)\times\BC^{(n'')}_L(\Bu'')$, since the cube $\BC^{(n)}_L(\Bu)$ is PI we have $\sigma(\BH^{(n)}_{\BC^{(n)}})=\sigma(\BH^{(n')}_{\BC^{(n')}_{\ell}(\Bx)})+\sigma(\BH^{(n'')}_{\BC^{(n'')}_L(\Bu'')})$, hence 
\[
\dist(E,\sigma(\BH^{(n)}_{\BC^{(n)}}))\leq |E-\mu_j-\eta|<\ee^{-\ell^{\beta}}.
\]
Thus $\BC^{(n)}$ is $E$-R. The same arguments shows that case (ii) arises when (b) occurs.
\end{proof}
%-------------------%
%:lem 10
%-------------------%
\begin{lemma} \label{lem:NDRoNS}
Let $E\in I$ and $\BC_{L_k}^{(n)}(\Bu)$ be a PI cube. Assume that
$\BC_{L_k}^{(n)}(\Bu)$ is $(E,m)$-NT and $E$-HNR.
Then $\BC^{(n)}_{L_k}(\Bu)$ is $(E,m)$-NS.
\end{lemma}
%-------------------%

%-------------------%
\begin{proof}
Let $\BC^{(n')}_{L_k}(\Bu')\times\BC^{(n'')}_{L_k}(\Bu'')$ be the decomposition of the PI cube $\BC^{(n)}_{L_k}(\Bu)$.
Let $\{\lambda_i,\varphi_i\}$ and $\{\mu_j,\phi_j\}$ be the eigenvalues and corresponding eigenvectors of $\BH_{\BC_{L_k}^{(n')}(\Bu')}^{(n')}$ and $\BH_{\BC_{L_k}^{(n'')}(\Bu'')}^{(n'')}$ respectively. Then we can choose the eigenvectors $\BPsi_{ij}$ and corresponding eigenvalues $E_{ij}$ of $\BH_{\BC^{(n)}_{L_k}(\Bu)}(\omega)$ as follows.
\[
\BPsi_{ij}=\varphi_i\otimes\phi_j,\qquad E_{ij}=\lambda_i+\mu_j.
\]
By the assumed  $E$-HNR property of the cube $\BC^{(n)}_{L_k}(\Bu)$, for all eigenvalues $\lambda_i$ one has $\BC^{(n'')}_L(\Bu'')$ is  $E-\lambda_i$-CNR. Next,  by assumption of $(E,m)$-NT, $\BC^{(n'')}_{L_k}(\Bu'')$ does not contain any pair of separable $(E-\lambda_i,m)$-S cubes of radius $L_{k-1}$ therefore by Lemma \ref{lem:MPI} $M(\BC^{(n)}_{L_{k+1}}(\Bu),E-\lambda_i)<\kappa(n)+2$ and Lemma \ref{lem:CNR.NS} implies that it is also $(E-\lambda_i,m)$-NS, yielding 
\begin{equation}\label{eq:lem.NDRoNS.1}
\max_{\{\lambda_i\}}
\max_{\Bv''\in \partial^- \BC_{L_k}^{(n'')}(\Bu'')}
\left|\BG^{(n'')}(\Bu'',\Bv'';E-\lambda_i)\right|\leq\ee^{-\gamma(m,L_k,n'')L_k}.
\end{equation}
The same analysis for $\BC^{(n')}_L(\Bu')$ also gives
\begin{equation}
\label{eq:lem.NDRoNS.2}
\max_{\{ \mu_j\}}
      \max_{\Bv'\in \partial^-\BC_{L_k}^{(n')}(\Bu')}
      \left|\BG^{(n')}(\Bu',\Bv';E-\mu_j)\right|\leq\ee^{-\gamma(m,L_k,n')L_k}.
\end{equation}
For any $\Bv\in\partial^-\BC^{(n)}_{L_k}(\Bu)$, $|\Bu-\Bv|=L_k$, thus either $|\Bv'-\Bu'|=L_k$ or $|\Bv''-\Bu''|=L_k$. Consider first the latter case. Equation \eqref{eq:lem.NDRoNS.1} applies and we get
\begin{gather*}
\left|\BG^{(n)}(\Bu,\Bv;E)\right| =\left|\sum_{i,j} \frac{\varphi_i(\Bu')\varphi_i(\Bv')\phi_j(\Bu'')\phi_j(\Bv'')}{E-\lambda_i-\mu_j}\right|\\
\leq \sum_i \left|\varphi_i(\Bu')\varphi_i(\Bv')\right|\cdot \left|\BG^{(n)}(\Bu'',\Bv'';E-\lambda_i)\right|\\
\leq (2L_k+1)^{(n-1)d}
\max_{\{ \lambda_i \}}\;
\max_{\Bv''\in\partial\BC^{(n)}_{L_k}(\Bu'') }
\left|\BG^{(n)}(\Bu'',\Bv'';E-\lambda_i)\right|, \tag{ since $\|\varphi\|_{\infty}\leq 1$}\\
\leq (2L_k+1)^{(n-1)d}\cdot\ee^{-\gamma(m,L_k,n-1)L_k}\\
=\ee^{-[\gamma(m,L_k,n-1)-L_k^{-1}\ln(2L_k+1)^{(n-1)d}]L_k}.
\end{gather*}
But by  Definition \eqref{eq:gamma}:
\[
\gamma(m,L_k,n)=m(1+L_k^{-1/8})^{N-n+1},
\]
with $m=(14N^N+6Nd)L_0^{-1/2}$.
For $2\leq n\leq N$,
\[
\gamma(m,L_k,n-1)-\gamma(m,L_k,n)>L_k^{-1}\ln(2L_k+1)^{(n-1)d}.
\]
Indeed, setting $C_1=14N^N+6Nd$,
\begin{align*}
\gamma(m,L_k,n-1)-\gamma(m,L_k,n)&=mL_k^{-1/8}(1+L_k^{-1/8})^{N-n+1}\\
&=C_1L_0^{-1/2}L_k^{-1/8}(1+L_k^{-1/8})^{N-n+1}>C_1L_k^{-5/8},
\end{align*}
and for $L_0$  sufficiently large, hence $L_k$,
\[
L_k^{-1}\ln(2L_k+1)^{(n-1)d}\leq L_k^{-1}(n-1)d(3L_k)^{3/8}\leq C_1 L_k^{-5/8}.
\]
Thus, $\BC^{(n)}_{L_k}(\Bu)$ is $(E,m)$-NS. Finally, the case $|\Bu'-\Bv'|=L_k$ is similar.
\end{proof}
%-------------------%

%-------------------%
%:lem 11
%-------------------%
\begin{lemma}\label{lem:T.estimate}
Let $2\leq n\leq N$ and assume property $\dsn{k,n'}$ for any $1\leq n'<n$.
Then for any  PI cube $\BC_{L_{k+1}}^{(n)}(\By)$ one has
\begin{equation}\label{eq:C.is.T}
\DP\bigl\{\exists E\in I, \BC_{L_{k+1}}^{(n)}(\By)\text{ is $(E,m)$-T}\bigr\}\leq \frac{1}{2}L_{k+1}^{-4p\,4^{N-n}}.
\end{equation}
\end{lemma}
%-------------------%

%-------------------%
\begin{proof}
Consider a PI cube  $\BC_{L_{k+1}}^{(n)}(\By)=\BC_{L_{k+1}}^{(n')}(\By')\times\BC_{L_{k+1}}^{(n'')}(\By'')$.
By definition \ref{def:tunnelling}, we have that the event
\[
\left\{\exists E\in I:\BC_{L_{k+1}}^{(n)}(\By)\text{ is $(E,m)$-T}\right\},
\]
is contained in the union
\[
\left\{\exists E\in I: \BC^{(n)}_{L_{k+1}}(\By)\text{ is $(E,m)$-RT}\right\}\cup\left\{\exists E\in I: \BC^{(n)}_{L_{k+1}}(\By)\text{ is $(E,m)$-LT}\right\}.
\]
Now, since $E\in I$ and $\mu_j\geq 0$ we have $ E-\mu_j\leq E^*$. So for any $j$, $E-\mu_j\in I$.
Further using property $\dsn{k,n'}$ we have
\begin{align*}
\prob{\text{$\exists E\in I$, $\BC^{(n)}_{L_{k+1}}(\By)$ is $(E,m)$-RT}}
&\leq \frac{|\BC^{(n')}_{L_{k+1}}(\By')|^{2}}{2}|\BC^{(n'')}_{L_{k+1}}(\By'')|L_k^{-2p4^{N-n'}}\\
&\leq C(n,N,d)L_{k+1}^{-2p\,\frac{4^{N-(n-1)}}{\alpha}+3(n-1)d}.
\end{align*}
A similar argument also shows that
\[
\prob{\text{$\exists E\in I$, $\BC^{(n)}_{L_{k+1}}(\By)$ is $(E,m)$-LT}}
\leq C(n,N,d)L_{k+1}^{-2p\,\frac{4^{N-(n-1)}}{\alpha}+3(n-1)d},
\]
so that
\[
\prob{\exists E\in I: \BC^{(n)}_{L_{k+1}}(\By)\text{ is $(E,m)$-T}}
\leq C(n,N,d)L_{k+1}^{-2p\,\frac{4^{N-(n-1)}}{\alpha}+3(n-1)d}.
\]
The assertion follows by observing that $2p\,4^{N-(n-1)}/\alpha \, -3(n-1)d>4p\,4^{N-n}$ for $\alpha=3/2$ provided
$L_0$ is large enough and  $p>4\alpha Nd=6Nd$.
\end{proof}
%-------------------%

%-------------------%
%:thm 6
%-------------------%
\begin{theorem}\label{thm:partially.interactive}
Let $1\leq n\leq N$. There exists $L_1^*=L_1^*(N,d)>0$ such that if $L_0\geq L_1^*$ and if for $k\geq 0$ $\dsn{k,n'}$ holds true for any $1\leq n'<n$, then $\dskonn$ holds true
for any pair of separable PI cubes $\BC_{L_{k+1}}^{(n)}(\Bx)$ and $\BC_{L_{k+1}}^{(n)}(\By)$.
\end{theorem}
%-------------------%
\begin{proof}
Let $\BC_{L_{k+1}}^{(n)}(\Bx)$ and $\BC_{L_{k+1}}^{(n)}(\By)$ be two separable PI-cubes. Consider the events:
\begin{align*}
\rB_{k+1}
&=\bigl\{\exists\,E\in I:\BC_{L_{k+1}}^{(n)}(\Bx)\text{ and $\BC_{L_{k+1}}^{(n)}(\By)$ are $(E,m)$-S}\bigr\},\\
\rR
&=\bigl\{\exists\,E\in I:\text{neither $\BC_{L_{k+1}}^{(n)}(\Bx)$ nor $\BC_{L_{k+1}}^{(n)}(\By)$ is $E$-HNR}\bigr\},\\
\rT_{\Bx}
&=\bigl\{\exists E\in I: \BC_{L_{k+1}}^{(n)}(\Bx)\text{ is $(E,m)$-T}\bigr\},\\
\rT_{\By}
&=\bigl\{\exists E\in I: \BC_{L_{k+1}}^{(n)}(\By)\text{ is $(E,m)$-T}\bigr\}.
\end{align*}
If $\omega\in\rB_{k+1}\setminus\rR$, then $\forall E\in I$, $\BC_{L_{k+1}}^{(n)}(\Bx)$ or $\BC_{L_{k+1}}^{(n)}(\By)$ is $E$-HNR. If $\BC_{L_{k+1}}^{(n)}(\By)$ is $E$-HNR, then it must be $(E,m)$-T: otherwise it would have been $(E,m)$-NS by Lemma~\ref{lem:NDRoNS}. Similarly, if $\BC_{L_{k+1}}^{(n)}(\Bx)$ is $E$-HNR, then it must be $(E,m)$-T. This implies that
\[
\rB_{k+1}\subset\rR\cup\rT_{\Bx}\cup\rT_{\By}.
\]
Therefore,
\begin{align*}
\DP\left\{\rB_{k+1}\right\}&\leq\DP\left\{\rR\right\}+\DP\{\rT_{\Bx}\}+\DP\{\rT_{\By}\}\\
&\leq \prob{\rR}+\frac{1}{2}L_{k+1}^{-4p\,4^{N-n}}+\frac{1}{2}L_{k+1}^{-4p\,4^{N-n}}\\
\end{align*}
where we used  \eqref{eq:C.is.T} to estimate $\DP\{\rT_{\Bx}\}$ and $\DP\{\rT_{\By}\}$. Next by combining Theorem \ref{thm:Wegner} and Lemma \ref{lem:HNR} we obtain as in corollary \ref{cor:Wegner} that $\prob{\rR}\leq L_{k+1}^{-4^N\,p}$. Finally 
\begin{equation}\label{eq:bound.PI}
\prob{\rB_{k+1}}\leq L_{k+1}^{-4^Np}+L_{k+1}^{-4p4^{N-n}}<L_{k+1}^{-2p4^{N-n}}.
\end{equation}
\end{proof}
%-------------------%

%-------------------%
%:lem 12
%-------------------%
\begin{lemma}\label{lem:MND}
With the above notations, assume that $\dsn{k-1,n'}$ holds true for all $1\leq n'<n$ then
\begin{equation}\label{eq:MND}
\DP\left\{M_{\pai}(\BC_{L_{k+1}}^{(n)}(\Bu),I)\geq \kappa(n)+ 2\right\}\leq \frac{3^{2nd}}{2} L_{k+1}^{2nd}\left(L_k^{-4^Np}+L_k^{-4p\,4^{N-n}}\right).
\end{equation}
\end{lemma}
%-------------------%

%-------------------%
\begin{proof}
Suppose that $M_{\pai}(\BC^{(n)}_{L_{k+1}}(\Bu),I)\geq \kappa(n)+2$, then by Lemma \ref{lem:MPI}, $M_{\pai}^{\sep}(\BC^{(n)}_{L_{k+1}}(\Bu), I)\geq 2$, i.e., there are at least two separable $(E,m)$-singular PI cubes $\BC^{(n)}_{L_k}(\Bu^{(j_1)})$, $\BC^{(n)}_{L_k}(\Bu^{(j_2)})$ inside $\BC^{(n)}_{L_{k+1}}(\Bu)$.
The  number of possible pairs of centers $\{\Bu^{(j_1)},\Bu^{(j_2)}\}$ such that
\[
\BC_{L_k}^{(n)}(\Bu^{(j_1)}),\,\BC_{L_k}^{(n)}(\Bu^{(j_2)})\subset\BC_{L_{k+1}}^{(n)}(\Bu)
\]
 is bounded by $\frac{3^{2nd}}{2}L_{k+1}^{2nd}$. Then, setting 
\[
\rB_k=\{\text{$\exists E\in I$, $\BC^{(n)}_{L_k}(\Bu^{(j_1)})$, $\BC^{(n)}_{L_k}(\Bu^{(j_2)})$ are $(E,m)$-S}\},
\]
\[
\DP\left\{M_{\pai}^{\sep}(\BC_{L_{k+1}}^{(n)}(\Bu),I)\geq 2\right\}\leq\frac{3^{2nd}}{2}L_{k+1}^{2nd}\times\prob{\rB_k}
\]
with  $\prob{\rB_k}\leq L_k^{-4^Np}+L_k^{-4p\,4^{N-n}}$ by \eqref{eq:bound.PI}.
Here $\rB_k$ is defined as in Theorem \ref{thm:partially.interactive}.

\end{proof}
%-------------------%

%---------------------------------------------------------------------%
%:s.4.3
%---------------------------------------------------------------------%
\subsection{Pairs of fully interactive cubes}\label{ssec:FI.cubes}
Our aim now is to prove $\dskonn$ for a pair of separable fully interactive cubes $\BC_{L_{k+1}}^{(n)}(\Bx)$ and $\BC_{L_{k+1}}^{(n)}(\By)$.
The main result of this subsection is Theorem \ref{thm:fully.interactive}. We need the following preliminary result.

%-------------------%
%:lem 13
%-------------------%
\begin{lemma}\label{lem:MD}
Given $k\geq0$, assume that property $\dsknn$ holds true for all pairs of separable FI cubes. Then for any $\ell\geq 1$
\begin{equation}\label{eq:MD}
\DP\left\{M_{\fui}(\BC_{L_{k+1}}^{(n)}(\Bu),I)\geq 2\ell\right\}\leq C(n,N,d,\ell)L_k^{2\ell dn\alpha}L_k^{-2\ell p\,4^{N-n}}.
\end{equation}
\end{lemma}
%-------------------%

%-------------------%
\begin{proof}
Suppose there exist $2\ell$ pairwise separable, fully interactive cubes $\BC_{L_k}^{(n)}(\Bu^{(j)})$ $\subset\BC_{L_{k+1}}^{(n)}(\Bu)$, $1\leq j\leq 2\ell$. Then, by Lemma \ref{lem:disjointness}, for any pair $\BC_{L_k}^{(n)}(\Bu^{(2i-1)})$, $\BC_{L_k}^{(n)}(\Bu^{(2i)})$, the corresponding random Hamiltonians $\BH_{\BC_{L_k}^{(n)}(\Bu^{(2i-1)})}^{(n)}$ and $\BH^{(n)}_{\BC_{L_k}^{(n)}(\Bu^{(2i)})}$ are independent, and so are their spectra and their Green functions. For $i=1,\dots,\ell$ we consider the events:
\[
\rA_i=\left\{\exists\,E\in I:\BC_{L_k}^{(n)}(\Bu^{(2i-1)})\text{ and $\BC_{L_k}^{(n)}(\Bu^{(2i)})$ are $(E,m)$-S}\right\}.
\]
Then by assumption $\dsknn$, we have, for $i=1,\dots,\ell$,
\begin{equation}
\DP\left\{\rA_i\right\}\leq L_k^{-2p\,4^{N-n}},
\end{equation}
and, by independence of events $\rA_1,\dots,\rA_{\ell}$,
\begin{equation}
\DP\Bigl\{\bigcap_{1\leq i\leq\ell}\rA_i\Bigr\}=\prod_{i=1}^{\ell}\DP(\rA_i)\leq\bigl(L_k^{-2p\,4^{N-n}}\bigr)^{\ell}.
\end{equation}
To complete the proof, note that the total number of different families of $2\ell$ cubes $\BC_{L_k}^{(n)}(\Bu^{(j)})\subset\BC_{L_{k+1}}^{(n)}(\Bu)$, $1\leq j\leq 2\ell$, is bounded by
\[
\frac{1}{(2\ell)!}\left|\BC_{L_{k+1}}^{(n)}(\Bu)\right|^{2\ell}\leq C(n,N,\ell,d)L_{k}^{2\ell dn\alpha}.\qedhere
\]
\end{proof}
%-------------------%

%-------------------%
%:thm 7
%-------------------%
\begin{theorem}\label{thm:fully.interactive}
Let $1\leq n\leq N$. There exists $L_2^*=L_2^*(N,d)>0$ such that if $L_0\geq L_2^*$ and if for $k\geq 0$
\begin{enumerate}[\rm(i)]
\item
$\dsn{k-1,n'}$ for all $1\leq n'<n$ holds true,
\item
$\dsn{k,n}$ holds true for all pairs of FI cubes,
\end{enumerate}
then $\dskonn$ holds true
for any pair of separable FI cubes $\BC_{L_{k+1}}^{(n)}(\Bx)$ and $\BC_{L_{k+1}}^{(n)}(\By)$.
\end{theorem}
%-------------------%
Above, we used the convention that $\dsn{-1,n}$ means no assumption.

%-------------------%
\begin{proof}
Consider a pair of separable FI cubes $\BC_{L_{k+1}}^{(n)}(\Bx)$, $\BC_{L_{k+1}}^{(n)}(\By)$ and set $J=\kappa(n)+5$. Define
\begin{align*}
\rB_{k+1}
&=\left\{\exists\,E\in I:\BC_{L_{k+1}}^{(n)}(\Bx)\text{ and $\BC_{L_{k+1}}^{(n)}(\By)$ are $(E,m)$-S}\right\},\\
\Sigma
&=\left\{\exists\,E\in I:\text{neither $\BC_{L_{k+1}}^{(n)}(\Bx)$ nor $\BC_{L_{k+1}}^{(n)}(\By)$ is $E$-CNR}\right\},\\
\rS_{\Bx}
&=\left\{\exists\,E\in I:M(\BC_{L_{k+1}}^{(n)}(\Bx);E)\geq J+1\right\},\\
\rS_{\By}
&=\left\{\exists\,E\in I:M(\BC_{L_{k+1}}^{(n)}(\By),E)\geq J+1\right\}.
\end{align*}
Let $\omega\in\rB_{k+1}$. If $\omega\notin\Sigma\cup\rS_{\Bx}$, then $\forall E\in I$ either $\BC_{L_{k+1}}^{(n)}(\Bx)$ or $\BC_{L_{k+1}}^{(n)}(\By)$ is $E$-CNR and $M(\BC_{L_{k+1}}^{(n)}(\Bx),E)\leq J$. The cube $\BC_{L_{k+1}}^{(n)}(\Bx)$ cannot be $E$-CNR: indeed, by Lemma \ref{lem:CNR.NS} it would be $(E,m)$-NS. So the cube $\BC_{L_{k+1}}^{(n)}(\By)$ is $E$-CNR and $(E,m)$-S. This implies again by Lemma \ref{lem:CNR.NS} that
\[
M(\BC_{L_{k+1}}^{(n)}(\By),E)\geq J+1.
\]
Therefore $\omega\in\rS_{\By}$, so that $\rB_{k+1}\subset\Sigma\cup\rS_{\Bx}\cup\rS_{\By}$, hence
\[
\DP\left\{\rB_{k+1}\right\}
\leq\DP\{\Sigma\}+\DP\{\rS_{\Bx}\}+\DP\{\rS_{\By}\}
\]
and $\prob{\Sigma}\leq L_{k+1}^{-4^N\,p}$ by corollary \ref{cor:Wegner}.
Now let us estimate $\DP\{\rS_{\Bx}\}$ and similarly $\DP\{\rS_{\By}\}$. Since 
\[
M_{\pai}(\BC_{L_{k+1}}^{(n)}(\Bx),E)+M_{\fui}(\BC_{L_{k+1}}^{(n)}(\Bx),E)\geq M(\BC_{L_{k+1}}^{(n)}(\Bx),E),
\] 
the inequality $M(\BC_{L_{k+1}}^{(n)}(\Bx),E)\geq\kappa(n)+ 6$, implies that either $M_{\pai}(\BC_{L_{k+1}}^{(n)}(\Bx),E)\geq\kappa(n)+ 2$, or $M_{\fui}(\BC_{L_{k+1}}^{(n)}(\Bx),E)\geq 4$. Therefore, by Lemma \ref{lem:MND} and Lemma \ref{lem:MD} (with $\ell=2$),
\begin{align*}
\DP\{\rS_{\Bx}\}&\leq\DP\left\{\exists\,E\in I:M_{\pai}(\BC_{L_{k+1}}^{(n)}(\Bx),E)\geq \kappa(n)+2\right\}\\
&\quad+\DP\left\{\exists\,E\in I:M_{\fui}(\BC_{L_{k+1}}^{(n)}(\Bx),E)\geq 4\right\}\\
&\leq\frac{3^{2nd}}{2}L_{k+1}^{2nd}(L_k^{-4^Np}+L_k^{-4p\,4^{N-n}})+C'(n,N,d)L_{k+1}^{4 dn-\frac{4 p}{\alpha}4^{N-n}}\\
&\leq C''(n,N,d)\left(L_{k+1}^{-\frac{4^Np}{\alpha}+2nd}+L_{k+1}^{-\frac{4p}{\alpha}4^{N-n}+2nd}+L_{k+1}^{-\frac{4p}{\alpha}4^{N-n}+4nd}\right)\\
&\leq C'''(n,N,d) L_{k+1}^{-\frac{4p}{\alpha}4^{N-n}+4nd}\tag{$\alpha=3/2$}\\
&\leq\frac{1}{4}L_{k+1}^{-2p\,4^{N-n}},
\end{align*}
where we used $p>4\alpha Nd=6Nd$. Finally 
\[ 
\prob{\rB_{k+1}}\leq L_{k+1}^{-4^Np}+\frac{1}{2}L_{k+1}^{-2p4^{N-n}}<L_{k+1}^{-2p4^{N-n}}.
\]
\end{proof}
%-------------------%

%---------------------------------------------------------------------%
%:s.4.4
%---------------------------------------------------------------------%
\subsection{Mixed pairs of cubes}\label{ssec:mixed.S}
Now it remains only to derive $\dskonn$ in case (III), i.e., for pairs of $n$-particle cubes where one is PI while the other is FI.

%-------------------%
%:thm 8
%-------------------%
\begin{theorem}\label{thm:mixed}
Let $1\leq n\leq N$. There exists $L_3^*=L_3^*(N,d)>0$ such that if $L_0\geq L_3^*(N,d)$ and if for $k\geq 0$,
\begin{enumerate}[\rm(i)]
\item
$\dsn{k-1,n'}$ holds true for all $1\leq n'<n$,
\item
$\dsn{k,n'}$ holds true for all $1\leq n'<n$ and
\item
$\dsknn$ holds true for all pairs of FI cubes,
\end{enumerate}
then $\dskonn$ holds true for any pair of separable cubes $\BC_{L_{k+1}}^{(n)}(\Bx)$, $\BC_{L_{k+1}}^{(n)}(\By)$ where one is PI while the other is FI.
\end{theorem}
%-------------------%

%-------------------%
\begin{proof}
Consider a pair of separable $n$-particle cubes $\BC_{L_{k+1}}^{(n)}(\Bx)$, $\BC_{L_{k+1}}^{(n)}(\By)$ and suppose that $\BC_{L_{k+1}}^{(n)}(\Bx)$ is PI while $\BC_{L_{k+1}}^{(n)}(\By)$ is FI. Set $J=\kappa(n)+5$ and introduce the events
\begin{align*}
\rB_{k+1}
&=\left\{\exists\,E\in I:\BC_{L_{k+1}}^{(n)}(\Bx)\text{ and $\BC_{L_{k+1}}^{(n)}(\By)$ are $(E,m)$-S}\right\},\\
\Sigma
&=\left\{\exists\,E\in I: \BC_{L_{k+1}}^{(n)}(\Bx)\ \text{is not $E$-HNR and}\ \BC_{L_{k+1}}^{(n)}(\By)\text{ is not $E$-CNR}\right\},\\
\rT_{\Bx}
&=\left\{\exists E\in I: \BC_{L_{k+1}}^{(n)}(\Bx)\text{ is $(E,m)$-T}\right\},\\
\rS_{\By}
&=\left\{ \exists\,E\in I:M(\BC_{L_{k+1}}^{(n)}(\By),E)\geq J+1\right\}.
\end{align*}
Let $\omega\in \rB_{k+1}\setminus (\Sigma\cup \rT_{\Bx})$, then for all $E\in I$ either $\BC_{L_{k+1}}^{(n)}(\Bx)$ is $E$-HNR or $\BC_{L_{k+1}}^{(n)}(\By)$ is $E$-CNR and  $\BC_{L_{k+1}}^{(n)}(\Bx)$ is $(E,m)$-NT. The cube $\BC_{L_{k+1}}^{(n)}(\Bx)$ cannot be $E$-HNR. Indeed, by Lemma \ref{lem:NDRoNS} it would have been $(E,m)$-NS. Thus the cube $\BC_{L_{k+1}}^{(n)}(\By)$ is $E$-CNR, so by Lemma \ref{lem:CNR.NS}, $M(\BC_{L_{k+1}}^{(n)}(\By);E)\geq J+1$: otherwise $\BC_{L_{k+1}}^{(n)}(\By)$ would be $(E,m)$-NS. Therefore $\omega\in \rS_{\By}$. Consequently,
\[
\rB_{k+1}\subset \Sigma \cup \rT_{\Bx}\cup\rS_{\By}.
\]
Recall that the probabilities $\DP\{\rT_{\Bx}\}$ and $\DP\{\rS_{\By}\}$ have already been estimated in Sections \ref{ssec:PI.cubes} and \ref{ssec:FI.cubes}. We therefore obtain
\begin{align*}
\DP\left\{\rB_{k+1}\right\}&\leq \DP\{\Sigma\}+\DP\{\rT_{\Bx}\}+\DP\{\rS_{\By}\}\\
&\leq L_{k+1}^{-4^Np}+ \frac{1}{2}L_{k+1}^{-4p\,4^{N-n}}+ \frac{1}{4}L_{k+1}^{-2p\,4^{N-n}}\leq L_{k+1}^{-2p\,4^{N-n}}.\qedhere
\end{align*}
\end{proof}
%-------------------%

%---------------------------------------------------------------------%
%:s.4.5
%---------------------------------------------------------------------%
%-------------------%

\subsection{Conclusion}
\begin{theorem}\label{thm:DS.k.N}
Let $1\leq n\leq N$ and  $\BH^{(n)}(\omega)=-\BDelta+\sum_{j=1}^nV(x_j,\omega)+\BU$, where $\BU$, $V$ satisfy
 $\condI$ and $\condP$  respectively. Then for any $p>6Nd$ there exist
 $m>0$ and $E^*>E_0^{(N)}$  such that $\dsknn$ holds true for all $k\geq 0$ provided $L_0$ is large enough.
\end{theorem}
\begin{proof}
We show by induction on $n$ that $\dsknn$ holds true for all $k\geq 0$. Notice that the parameters $m, E^*>0$ are given by Theorem \ref{thm:initial.estimate}.  For $n=1$ and owing to the Wegner estimates Corollary \ref{cor:Wegner} for Log-H\"older continuous distributions, property $\dsn{k,1}$ holds true for all $k\geq0$ by  single particle localization theory (\cites{DK89,K08}). Now suppose that for all $1\leq n'<n$, $\dsn{k,n'}$ holds true for all	 $k\geq 0$, we aim to prove that $\dsknn$ holds true for all $k\geq 0$. For $k=0$, $\dsn{0,n}$ is valid using Theorem \ref{thm:initial.estimate}. Next, suppose that $\dsn{k',n}$ holds true for all $k'<k$, then combining this last assumption with $\dsn{k,n'}$ above, one can conclude that
\begin{enumerate}
\item[\rm(i)] $\dsknn$ holds true for all $k\geq 0$ and for all pairs of PI cubes using Theorem \ref{thm:partially.interactive},
\item[\rm(ii)] $\dsknn$ holds true for all $k\geq 0$ and for all pairs of FI cubes using Theorem \ref{thm:fully.interactive},
\item[\rm(iii)] $\dsknn$ holds true for all $k\geq 0$ and for all pairs of MI cubes using Theorem \ref{thm:mixed}.
\end{enumerate}
 Hence Theorem \ref{thm:DS.k.N} is proven.
\end{proof}

%-----------------------------------------------------------------%
%:s.5
%-----------------------------------------------------------------%

\section{Proof of the results}\label{sec:proof.results}

\subsection{Proof of Theorem \ref{thm:inf.spectrum}}
Assumption $\condI$ implies that the interaction $\BU$ is non-negative and $\condP$ implies that the random potential is non-negative almost surely. So
$\sigma(\BH^{(n)}(\omega))\subset [0,+\infty)$ almost surely.
 It remains to prove that $[0,4nd]\subset \sigma(\BH^{(n)}(\omega))$ \emph{a.s.}
Let $k,m\in\DN^*$. Define $I_{k,m}=[0,\frac{1}{km}]$ and 
\[
\rB_{k,m}=\left\{\Bx\in\DZ^{nd}: \min_{i\neq j}|x_i-x_j|>r_0+2km\right\}
\]
where $r_0$ is the range of the interaction $\BU$. For $k,m\in\DN^*$, introduce the following sequence $\{\Bx^{\ell,k,m}\}_{\ell\in\DN^*}$ defined by 
\[
\Bx^{\ell,k,m}=C_{k,m}(C_{k,m}\ell+1,C_{k,m}\ell+2,\cdots,C_{k,m}\ell+nd)\in\DZ^{nd},
\]
where $C_{k,m}=r_0+2km+Nd+1$. Obviously $\Bx^{\ell,k,m}\in \rB_{k,m}$ for any $\ell\in\DN^*$. Using the identification $\DZ^{nd}\cong(\DZ^d)^n$, $\Bx^{\ell,k,m}$ write as $\Bx^{\ell,k,m}=C_{k,m}(\Bx_1,\ldots,\Bx_n)$ with $\Bx_i=((i-1)d+1+C_{k,m}\ell,\ldots, id+C_{k,m}\ell)$, $i=1,\ldots,n$.

We claim that for $\ell\neq \ell'$, $\varPi\BC^{(n)}_{km}(\Bx^{\ell,k,m})\cap\varPi\BC^{(n)}_{km}(\Bx^{\ell',k,m})=\emptyset$. Indeed 
\begin{align*}
\dist\left(\varPi\BC^{(n)}_{km}(\Bx^{\ell,k,m}),\varPi\BC^{(n)}_{km}(\Bx^{\ell',k,m})\right)&=
\min_{1\leq i,j\leq n}\dist\left(\BC^{(n)}_{km}(\Bx^{\ell,k,m}_i),\BC^{(n)}_{km}(\Bx^{\ell',k,m}_j)\right)\\
&\geq \min_{i,j}\left|\Bx^{\ell,k,m}_i-\Bx^{\ell',k,m}_j\right|-2km.
\end{align*}
For any $i,j\in\{1,\ldots,n\}$,
\begin{align*}
\left|\Bx^{\ell,k,m}_i-\Bx^{\ell',k,m}_j\right|&=C_{k,m}|(i-j)d+C_{k,m}(\ell-\ell')|\\
&\geq C_{k,m}^2|\ell-\ell'|-C_{k,m}|i-j|d.
\end{align*}
Thus 
\begin{align*}
\dist\left(\varPi\BC^{(n)}_{km}(\Bx^{\ell,k,m}),\varPi\BC^{(n)}_{km}(\Bx^{\ell',k,m})\right)&\geq C_{k,m}^2-C_{k,m}\max_{i,j}|i-j|d-2km\\
&\geq C_{k,m}^2-C_{k,m}Nd-2km\\
&=C_{k,m}(C_{k,m}-Nd)-2km>0.
\end{align*}
For any $\ell,k,m\in\DN^*$, define 
\[
\Omega_{\ell,k,m}(\Bx^{\ell,k,m})=\left\{\omega\in\Omega:\forall j=1,\ldots,n, \forall y_j\in C^{(1)}_{km}(x_j^{\ell,k,m}), V(y_j,\omega)\in I_{km}\right\}.
\]
We have that $\prob{\Omega_{\ell,k,m}(\Bx^{\ell,k,m})}=\mu(I_{km})^{n(2km+1)^d}$ and the latter quantity is positive since $0\in\supp\mu$. So $\sum_{\ell\geq1}\prob{\Omega_{\ell,k,m}(\Bx^{\ell,k,m})}=+\infty$. Since $\varPi\BC^{(n)}_{km}(\Bx^{\ell,k,m})\cap\varPi\BC^{(n)}_{km}(\Bx^{\ell',k,m})=\emptyset$, the corresponding events $\Omega_{\ell,k,m}(\Bx^{\ell,k,m})$ and $\Omega_{\ell,k,m}(\Bx^{\ell',k,m})$ are independent. Set
\[
\Omega_{k,m}=\left\{\omega:\Omega_{\ell,k,m}(\Bx^{\ell,k,m})\text{ occurs for infinitely many $\ell\geq 1$}\right\},
\]
the Borel-Cantelli Lemma implies that $\prob{\Omega_{k,m}}=1$. Note that if $\omega\in\Omega_{k,m}$, then $\exists \ell\geq 1$ such that $\omega\in\Omega^{\ell,k,m}(\Bx^{\ell,k,m})$. Define
\[
\Omega_{\infty}=\bigcap_{k\geq 1}\bigcap_{m\geq 1}\Omega_{k,m}.
\]
We have that $\prob{\Omega_{\infty}}=1$. Now, let $\omega\in\Omega_{\infty}$. For this $\omega$ we show that $[0,4nd]\subset \sigma(\BH^{(n)}(\omega))$. Since $\omega\in\Omega_{\infty}$, we have that for any $k,m\geq 1$, $\exists\ell^*\geq 1$ such that $\omega \in\Omega_{\ell^*,k,m}(\Bx^{\ell^*,k,m})$. Recall that $[0,4nd]=\sigma(-\BDelta)$. Let $E\in [0,4nd]$. There exists a Weyl sequence $\phi^E_m$ for $E$ and $-\BDelta$ with a bounded support, i.e., $\|\phi^E_m\|=1$, $\|((-\BDelta)-E)\phi^E_m\|\rightarrow 0 $ as $m\rightarrow \infty$ and $\supp\phi^E_m\subset \BC^{(n)}_{k_E m}(\Bzero)$ for some $k_E\in\DN^*$. The translated sequence $\widetilde{\phi}^E_m$ defined by $\widetilde{\phi}^E_m(\Bx)=\phi^E_m(\Bx-\Bx^{\ell^*,k_E,m})$ is also a Weyl sequence for $E$ and $(-\BDelta)$, i.e,
\begin{equation}\label{eq:Delta.phi.m}
\|((-\BDelta)-E)\widetilde{\phi}^E_m\| \tto{m\to\infty} 0,
\end{equation}
 and $\supp \widetilde{\phi}^E_m\subset \BC^{(n)}_{k_E m}(\Bx^{\ell^*,k_E,m})$. Further, for all $\By\in\DZ^{nd}$ one has 
\begin{equation}\label{eq:V.supp.phi}
\left|\BV(\By,\omega)\widetilde{\phi}^E_m(\By)\right|\leq \frac{n}{k_E m}|\widetilde{\phi}^E_m(\By)|.
\end{equation}
Indeed, both sides of \eqref{eq:V.supp.phi} vanish outside
$\BC^{(n)}_{k_E m}(\Bx^{\ell^*,k_E,m})$, while on  $\BC^{(n)}_{k_E m}(\Bx^{\ell^*,k_E,m})$ it follows directly from our choice
of  $\omega$. Therefore,
\begin{equation}\label{eq:norm.V.phi}
\left \| (\BV(\omega)) \widetilde{\phi}^E_m  \right\| \leq \frac{n}{k_E m} \|\widetilde{\phi}^E_m \|,
\end{equation}
while $\BU \widetilde{\phi}^E_m = 0$.
Collecting \eqref{eq:Delta.phi.m} and \eqref{eq:norm.V.phi},
we conclude that
\begin{align*}
\| (\BH(\omega)-E) \widetilde{\phi}_m\| \le
\|((-\BDelta)-E)\widetilde{\phi}_m\| + \| (\BV) \widetilde{\phi}_m\| \tto{m\to\infty} 0.
\end{align*}
Thus $\{ \widetilde{\phi}^E_m, m\in\DN^*\}$ is indeed a Weyl sequence for $\BH^{(n)}(\omega)$ and $E$. Therefore $E\in\sigma(\BH^{(n)}(\omega))$. Finally, since $\prob{\Omega_{\infty}}=1$, we conclude that $[0,4nd]\subset \sigma(\BH^{(n)}(\omega))$ almost surely. $\qed$

\subsection{Proof of Theorem \ref{thm:low.energy.exp.loc}} \label{sec:low.energy.exp.loc}
%----------------------------------------------------------------------%

Now we will derive from the results of Section \ref{sec:MP.induction} spectral exponential localization in the low energy regime. We will use the well-known fact (\cites{B68,K08,S83}) that almost every energy $E$ with respect to the spectral measure of $\BH^{(N)}(\omega)$ is a generalised eigenvalue of $\BH^{(N)}(\omega)$, i.e., there is a polynomially bounded solution  of the equation $\BH^{(N)}(\omega)\BPsi=E\BPsi$. It suffices therefore to prove that with probability one, the generalised eigenfunctions of $\BH^{(N)}(\omega)$ decay exponentially fast at infinity.
Let $E\in [E_0^{(N)},E^*]$ be a generalised eigenvalue of $\BH^{(N)}(\omega)$. Following essentially \cites{DK89,CS09a,FMSS85}, we will prove that if the bound $\dskn$ is satisfied for all $k\geq 0$, with some $m>0$, then
\[
\forall \tilde{\rho}\in(0,1):\quad \limsup_{|\Bx|\to \infty} \ln \frac{|\BPsi(\Bx,\omega)|}{|\Bx|}\leq -\tilde{\rho}m.
\]
Given $\Bu\in \DZ^{Nd} $ and an integer $k\geq 0$, set, using the notations of Lemma \ref{lem:separability},
\[
R(\Bu):=\max_{\ell=1,\ldots,\kappa(N)}|\Bu-\Bu^{(\ell)}|,\quad b_k(\Bu):=7N+R(\Bu)L_k^{-1},\quad M_k(\Bu)=\bigcup_{\ell=1}^{\kappa(N)} \BC^{(N)}_{7NL_k}(\Bu^{(\ell)})
\]
and define
\[
A_{k+1}(\Bu):= \BC^{(N)}_{bb_{k+1}L_{k+1}}(\Bu)\setminus \BC^{(N)}_{b_kL_k}(\Bu),
\]
where the parameter $b>0$ is to be chosen later. One can easily check that:
 \[
M_k(\Bu)\subset \BC^{(N)}_{b_kL_k}(\Bu).
\]
Moreover, if $\Bx \in A_{k+1}(\Bu)$, then the cubes $\BC^{(N)}_{L_k}(\Bx)$ and $\BC^{(N)}_{L_k}(\Bu)$ are separable by Lemma \ref{lem:separability}.
Define the event
\[
\Omega_k(\Bu):=\bigl\{\exists E\in [E_0^{(N)},E^*], \text{ and $\Bx\in A_{k+1}(\Bu)$: $\BC^{(N)}_{L_k}(\Bx)$ and $\BC^{(N)}_{L_k}(\Bu)$ are $(E,m)$-S}\Bigr\}.
\]
Property $\dskn$ implies that
\begin{align*}
\prob{\Omega_k}&\leq \prob{\exists E\in (-\infty,E^*], \text{ and $\Bx\in A_{k+1}(\Bu)$: $\BC^{(N)}_{L_k}(\Bx)$ and $\BC^{(N)}_{L_k}(\Bu)$ are $(E,m)$-S}}\\
&\leq (2bb_{k+1}L_{k+1}+1)^{Nd}L_k^{-2p}\\
& \leq (2bb_{k+1}+1)^{Nd}L_k^{-2p+\alpha Nd},
\end{align*}
since $p>(\alpha Nd+1)/2$ (in fact $p>6Nd$)  we have $\sum_{k=0}^{\infty} \prob{\Omega_k(\Bu)}<\infty$. Thus, setting  
\[
\Omega_{<\infty}:=\{ \forall \Bu\in\DZ^{Nd}, \Omega_k(\Bu) \text{ occurs finitely many times}\},
\]
by the Borel-Cantelli Lemma and the countability of $\DZ^d$, we have that
\[
\prob{\Omega_{<\infty}}=1.
\]
Therefore, it suffices to pick $\omega\in\Omega_{<\infty}$ and prove the exponential decay of any nonzero generalised eigenfunction $\BPsi$ of $\BH^{(N)}(\omega)$. Since $\BPsi$ is polynomially bounded, there exist $C,t\in(0,\infty)$ such that for all $\Bx\in \DZ^{Nd}$
\[
|\BPsi(\Bx,\omega)|\leq C(|\Bx|)^t.
\]
Since $\BPsi$ is not identically zero, there exists $\Bu\in \DZ^{Nd}$ such that $\BPsi(\Bu)\neq 0$. Let us show that there is an integer $k_1=k_1(\omega,E,\Bu)$ such that $\forall k\geq k_1$, the cube $\BC^{(n)}_{L_k}(\Bu)$ is $(E,m)$-S. Indeed, given an integer $k\geq 0$, assume that $\BC^{(n)}_{L_k}(\Bu)$ is $(E,m)$-NS. Then by the geometric resolvent inequality for the eigenfunction, combined with the definition of an $(E,m)$-NS cube (cf. \eqref{eq:singular}), we have
\begin{align*}
|\BPsi(\Bu)|&\leq C(N,d)L_k^{Nd-1}\ee^{-mL_k}\cdot \max_{\Bv:|\Bv-\Bu|\leq L_k+1}|\BPsi(\Bv)|\\
&\leq C'(N,d)L_k^{Nd-1}\ee^{-mL_k}(1+|\Bu|+L_k)^t\tto{L_k\to\infty} 0.
\end{align*}
This shows that if $\BC^{(n)}_{L_k}(\Bu)$ is $(E,m)$-NS for arbitrary large values of $L_k$ (i.e., for an infinite number of values $k$), then $|\BPsi(\Bu)|=0$, in contradiction with the definition of the point $\Bu$. So there is an integer $k_1=k_1(\omega,E,\Bu)<\infty$ such that $\forall k \geq k_1$ the cube $\BC^{(N)}_{L_k}(\Bu)$ is $(E,m)$-S. At the same time since $\omega\in \Omega_{<\infty}$ there exists $k_2=k_2(\omega,\Bu)$ such that if $k\geq k_2	 $, $\Omega_{k}(\Bu)$ does not occurs. We conclude that $\forall k\geq \max \{k_1,k_2\}$ for all $\Bx\in A_{k+1}(\Bu)$, $\BC^{(N)}_{L_k}(\Bx)$ is $(E,m)$-NS.
Let $\rho\in(0,1)$ and choose $b$ such that
\[
b>\frac{1+\rho}{1-\rho},
\]
so that obviously $\BC^{(N)}_{\frac{b_{k}L_{k}}{1-\rho}}(\Bu)\subset\BC^{(N)}_{\frac{bb_{k+1}L_{k+1}}{1+\rho}}(\Bu)$.
Define
\[
\tilde{A}_{k+1}(\Bu)=\BC^{(N)}_{\frac{bb_{k+1}L_{k+1}}{1+\rho}}(\Bu)\setminus \BC^{(N)}_{\frac{b_kL_k}{1-\rho}}(\Bu)\subset A_{k+1}(\Bu).
\]
Fix $\Bx\in\tilde{A}_{k+1}(\Bu)$.
\begin{enumerate}
\item
Since $|\Bx-\Bu|>\frac{b_kL_k}{1-\rho}$,
\begin{align*}
\dist\Bigl(\Bx,\partial^+ \BC^{(N)}_{b_kL_k}(\Bu)\Bigr)&=|\Bx-\Bu|-b_kL_k\\
& > |\Bx-\Bu|-(1-\rho)|\Bx-\Bu|\\
&=\rho|\Bx-\Bu|.
\end{align*}
\item
Since $|\Bx-\Bu|\leq \frac{bb_{k+1}L_{k+1}}{1+\rho}$,
\begin{align*}
\dist\Bigl(\Bx,\partial^-\BC^{(N)}_{bb_{k+1}L_{k+1}}(\Bu)\Bigr)&=bb_{k+1}L_{k+1}-|\Bx-\Bu|\\
&\geq (1+\rho)|\Bx-\Bu|-|\Bx-\Bu|\\
&=\rho|\Bx-\Bu|.
\end{align*}
\end{enumerate}
The interior boundary of the annulus $A_{k+1}(\Bu)$ is given by $\partial^-A_{k+1}(\Bu)=\partial^-\BC^{(N)}_{bb_{k+1}L_{k+1}}(\Bu)\cup\partial^+\BC^{(N)}_{b_kL_k}(\Bu)$. Thus
\begin{align*}
\dist\Bigl(\Bx,\partial^- A_{k+1}(\Bu)\Bigr)&=\min\Bigl[\dist\Bigl(\Bx,\partial^+\BC^{(N)}_{b_kL_k}(\Bu)\Bigr),\dist\Bigl(\Bx,\partial^-\BC^{(N)}_{bb_{k+1}L_{k+1}}(\Bu)\Bigr)\Bigr]\\
&\geq \rho |\Bx-\Bu|.
\end{align*}
We have that if $|\Bx-\Bu|>b_0L_0/(1-\rho)$ then there exists $k\geq 0$ such that $\Bx\in \tilde{A}_{k+1}(\Bu)$. 

Now, let $k\geq \max\{k_1,k_2\}$, so  the cube $\BC^{(N)}_{L_k}(\Bx)$ must be $(E,m)$-NS. Hence by the geometric resolvent inequality for eigenfunctions, we get
\begin{equation}\label{eq:exp.decay.eigenf}
|\BPsi(\Bx)|\leq C(N,d)L_k^{Nd-1}\ee^{-m(1+L_k^{-1/8})L_k}|\BPsi(\Bv_1)| \quad \text{with $\Bv_1\in\partial^+\BC^{(N)}_{L_k}(\Bx)$}.
\end{equation}

If $\Bx\in\tilde{A}_{k+1}(\Bu)$ with $k\geq \max\{k_1,k_2\}$, we  can iterate the bound \eqref{eq:exp.decay.eigenf} at least $(L_k+1)^{-1}\rho|\Bx-\Bu|$ times and, using the polynomial bound on $\BPsi$ to obtain
\begin{equation}\label{eq:eigenf.bound}
|\BPsi(\Bx)|\leq \Bigl[C(N,d)L_k^{Nd-1}\ee^{-m(1+L_k^{-1/8})L_k}\Bigr]^{(L_k+1)^{-1}\rho|\Bx-\Bu|}C(1+|\Bu|+ bL_{k+1})^t.
\end{equation}
We can conclude that given $\rho'$, $0<\rho'<1$, we can find $k_3\geq \max\{k_1,k_2\}$ such that if $k\geq k_3$ then 
\[
|\BPsi(\Bx)| \leq \ee^{-\rho\rho'm|\Bx-\Bu|},
\]
if $|\Bx-\Bu|>\frac{b_{k_3}L_{k_3}}{1-\rho}$.
Finally, we see that
\[
\limsup_{|\Bx|\to \infty}\frac{1}{|\Bx|}\ln|\BPsi(\Bx)|\leq -\rho\rho'm.
\]
%----------------------------%
%s.6
%----------------------------%
\subsection{Proof of Theorem \ref{thm:low.energy.dynamical.loc}}\label{sec:low.energy.dynamical.loc}

Dynamical localization via multi-scale analysis was proven initially by Germinet and De Bi\`evre \cite{GB98} and by Damanik and Stollmann (cf. \cites{DS01,St01}). In our work, we use a different approach, originally developed by Germinet and Klein \cite{GK01} in the framework of differential operators in $\DR^d$. This will allows  us to derive from the results of the multiparticle multi-scale analysis strong localization.

We need first to introduce necessary  notions and summarize in Theorem \ref{thm:spectral,densities} below some well-known results on expansions in generalized eigenfunctions for lattice Schr\"odinger operators which can be found, e.g., in the book \cite{K08}. Let $I\subset \DR$ and denote by $\nu(I)=P_I(H)$ the projection valued measure associated to $H$.  Further, given any pair of points  $n,m\in\DZ^D$, introduce a real valued Borel measure $\nu_{n,m}(\cdot)$ by
\[
\nu_{n,m}(I)=\left\langle \delta_n,\nu(I)\delta_m\right\rangle.
\]
Consider a sequence $\{\alpha_n\}_{n\in\DZ^D}$ with $\alpha_n>0$, $\sum \alpha_n=1$ and define a positive spectral measure $\rho(\cdot)$ by
\begin{equation}\label{eq:spectral,measure}
\rho(I)=\sum_{n\in\DZ^D}\alpha_n\nu_{n,n}(I).
\end{equation}
Observe that $\rho$ is a normalized Borel measure: $\rho(\DR)=1$. 
%----------------------------------------%
%: thm 9
%----------------------------------------%

%----------------------------------------%
\begin{theorem}\label{thm:spectral,densities}
Let $\rho$ be a spectral measure for $H=\Delta+W(x)$ acting in $\ell^{2}(\DZ^D)$. Then there exist measurable functions $F_{n,m}:\DR\rightarrow\DR$ such that
\begin{equation}\label{eq:spectral,densities}
\left\langle \delta_n,f(H)\delta_m\right\rangle=\int f(\lambda)F_{n,m}(\lambda)d\rho(\lambda)
\end{equation}
for any bounded measurable function $f:\DR\rightarrow\DR$. Furthermore, the functions\\ $\BPsi^{(m,\lambda)}: n\to F_{n,m}(\lambda)$ on $\DZ^D$ satisfy
\[
H\BPsi =\lambda \BPsi \quad  \text{for $\rho$-a.e. $\lambda$},
\]
and are polynomially bounded, i.e,
\[
|\BPsi(n)|\leq C(1+|n|)^t, \quad \text{for some $C>0$ and $t>0$}.
\]
\end{theorem}
%-----------------------------------------%

%-----------------------------------------%
\begin{proof}
See the proof of Proposition 7.4 in \cite{K08}.
\end{proof}
%------------------------------------------%
Before we turn to the proof of Theorem \ref{thm:low.energy.dynamical.loc}, we need to make the following observations.
For every bounded set $\BK\subset \DZ^{Nd}$ figuring in Theorem \ref{thm:low.energy.dynamical.loc} there exists $k_0>0$ such that $\BK\subset\BC^{(N)}_{L_{k_0}}(\Bzero)$. Now  for $j\geq k_0$, set
\[
\BM_j(\Bzero)=\BC^{(N)}_{(7N+1)L_{j+1}}(\Bzero)\setminus \BC^{(N)}_{(7N+1)L_j}(\Bzero).
\]
Observe that for any  $\By\in \BC^{(N)}_{L_{j}}(\Bzero)$, $\BC^{(N)}_{7NL_j}(\By)\subset \BC^{(N)}_{(7N+1)L_j}(\Bzero)$. Then, if $\Bx\in\BM_j(\Bzero)$ and $\By\in\BC^{(N)}_{L_j}(\Bzero)$, we have that $|\By-\Bx|>7NL_j$. Since $\diam (\varPi\By)\leq 2L_j$, it follows that $|\By-\Bx|>\diam(\varPi\By) + 3NL_j$. Thus using assertion (B) of Lemma \ref{lem:separability}, the cubes $\BC^{(N)}_{L_j}(\Bx)$ and $\BC^{(N)}_{L_j}(\By)$ are separable. We need  the following statement establishing the decay of the kernels in the Hilbert-Schmidt norm.

%----------------------------------------%
%: lem 14
%----------------------------------------%

%----------------------------------------%
\begin{lemma}\label{lem:decay,kernels}
Under assumptions $\condI$ and $\condP$. Then there exists an integer $k_1\geq 0$ such that for any bounded measurable function $f:\DR\rightarrow \DC$ all $j\geq k_1$, $\Bx\in\BM_j(\Bzero)$ and $\By\in\BC^{(N)}_{L_j}(\Bzero)$
\begin{equation}\label{eq:decay,kernels}
\esm{\sup_{\|f\|_{\infty}\leq 1}\left\|\delta_{\Bx}f(\BH)\BP_I(\BH)\delta_{\By}\right\|_{HS}^2}\leq \ee^{-mL_j/2}+L_j^{-2p},
\end{equation}
where $p>6Nd$.
\end{lemma}
%---------------------------------------%

%---------------------------------------%
\begin{proof}
Define
\[
B_j:=\left\{\forall \lambda\in I, \text{either $\BC^{(N)}_{L_j}(\Bx)$ or $\BC^{(N)}_{L_j}(\By)$ is $(\lambda,m)$-NS}\right\}
\]
Consider a bounded measurable function $f:\DR\rightarrow\DC$ and set $f_I=f\chi_I$, where $\chi_I$ is the indicator function of the interval $I$. We have:
\begin{align*}
\left\|\delta_{\Bx}f_I(\BH)\delta_{\By}\right\|_{HS}^2&=\sum_{\Bz,\Bv\in \DZ^{Nd}}\left|\left\langle\delta_{\Bx}f_I(\BH)\delta_{\By}\delta_{\Bz},\delta_{\Bv}\right\rangle\right|^2\\
&=\sum_{\Bz,\Bv\in\DZ^{Nd}}\left|\left\langle f_{I}(\BH)\delta_{\By}\delta_{\Bz},\delta_{\Bx}\delta_{\Bv}\right\rangle\right|^2\\
&=\left|\left\langle f_I(\BH)\delta_{\By},\delta_{\Bx}\right\rangle\right|^2\\
&=\left|\left\langle \delta_{\Bx},f_I(\BH)\delta_{\By}\right\rangle\right|^2.
\end{align*}
If $\omega\in B_j$ then either $\BC^{(N)}_{L_j}(\Bx)$ or $\BC^{(N)}_{L_j}(\By)$ is $(\lambda,m)$-NS for all $\lambda\in I$. Since
\[
|\left\langle \delta_{\Bx},f_I(\BH)\delta_{\By}\right\rangle|=|\left\langle \delta_{\By},\overline{f_I}(\BH)\delta_{\Bx}\right\rangle|,
\]
we can assume without loss of generality that $\BC^{(N)}_{L_j}(\Bx)$ is $(\lambda,m)$-NS. Now using Theorem \ref{thm:spectral,densities}, we get
\begin{align*}
\left\|\delta_{\Bx}f_I(\BH)\delta_{\By}\right\|_{HS}&\leq \left| \left\langle \delta_{\Bx},f_I(\BH)\delta_{\By}\right\rangle\right|\\
&\leq \int_I \left|f(\lambda)\right| \left|F_{\Bx,\By}(\lambda)\right|d\rho(\lambda)
\end{align*}
and the function $\BPsi:\Bx\rightarrow F_{\Bx,\By}(\lambda)$ is polynomially bounded for $\rho$ a.e. $\lambda$.
Next, the geometric resolvent inequality for generalised eigenfunctions  gives
\begin{align*}
|\BPsi(\Bx)|&\leq \left|\partial \BC^{(N)}_{L_j}(\Bx)\right|\ee^{-mL_j}|\BPsi(\Bx')|\\
&\leq C(N,d)L_j^{Nd-1}\ee^{-mL_j}(1+|\Bx|+L_j)^t\\
&\leq C(N,d)L_j^{\alpha t+Nd-1}\ee^{-mL_j}<\ee^{-mL_j/2}
\end{align*}
for $j\geq k_1$ with $k_1\geq0$ large enough.
Yielding
\begin{align*}
\left\|\delta_{\Bx}f_I(\BH)\delta_{\By}\right\|_{HS} & \leq \|f\|_{\infty}\rho(I)\ee^{-mL_j/2}\\
&\leq \|f\|_{\infty}\ee^{-mL_j/2}.
\end{align*}
For $\omega\in B_j^c$, we have
\[
\left\|\delta_{\Bx}f_I(\BH)\delta_{\By}\right\|_{HS}=\left|\left\langle \delta_{\Bx},f_I(\BH)\delta_{\By}\right\rangle\right|\leq \|f\|_{\infty}.
\]
Finally, we can conclude that
\begin{align*}
\esm{\sup_{\|f\|_{\infty}\leq 1}\left\| \delta_{\Bx}f_I(\BH)\delta_{\By}\right\|^2_{HS}}&\leq \ee^{-mL_j/2}\prob{B_j}+ \prob{B_j^c}\\
&\leq \ee^{-mL_j/2}+L_j^{-2p}.
\end{align*}
Above we used $\dskn$ to bound $\prob{B_j^c}$.
\end{proof}
%-----------------------------------------%

We are now ready to finish the proof of Theorem \ref{thm:low.energy.dynamical.loc}. Namely, let $\BK\subset\DZ^{Nd}$, $j\geq k_1$ as in Lemma \ref{lem:decay,kernels} and $p>6Nd$ the parameter appearing in the RHS of property $\dskn$. Set $s^*=\frac{2p}{\alpha}-Nd-1$. For any $s\in(0,s^*)$ and any bounded measurable function $f:\DR\rightarrow\DC$ we have
\begin{align*}
&\esm{\sup_{\|f\|_{\infty}\leq 1}\left\| |X|^{\frac{s}{2}}f_I(\BH)\Bone_{\BK}\right\|^2_{HS}}\\
&\leq \esm{\sup_{\|f\|_{\infty}
\leq 1}\left\|\Bone_{\BC^{(N)}_{(7N+1)L_{k_1}}(\Bzero)}|X|^{\frac{s}{2}}f_I(\BH)\Bone_{\BK}\right\|_{HS}^2}\\
&\qquad+\sum_{j\geq k_1}C_1(N,d)L_{j+1}^s \sum_{\Bx\in \BM_j(\Bzero)} \esm{
\sup_{\|f\|_{\infty}\leq 1}\left\| \delta_{\Bx}f_I(\BH)\Bone_{\BK}\right\|_{HS}^2} \\
&\leq \esm{\sup_{\|f\|_{\infty}
\leq 1}\sum_{\Bx\in\DZ^{Nd}}\left\|\Bone_{\BC^{(N)}_{(7N+1)L_{k_1}}(\Bzero)}|X|^{\frac{s}{2}}f_I(\BH)\Bone_{\BK}\delta_{\Bx}\right\|^2}\\
&\qquad+\sum_{j\geq k_1}C_1(N,d)L_{j+1}^s \sum_{\substack{\Bx\in \BM_j(\Bzero)\\ \By\in \BC^{(N)}_{L_{k_1}}(\Bzero)}} \esm{
\sup_{\|f\|_{\infty}\leq 1}\left\| \delta_{\Bx}f_I(\BH)\delta_{\By}\right\|_{HS}^2} \\
&\leq \esm{\sup_{\|f\|_{\infty}
\leq 1}\sum_{\Bx\in\BK}\left\|\Bone_{\BC^{(N)}_{(7N+1)L_{k_1}}(\Bzero)}|X|^{\frac{s}{2}}f_I(\BH)\delta_{\Bx}\right\|^2}\\
&\qquad+ C_2(N,d)L_{k_1}^{Nd}\sum_{j\geq k_1} L_j^{\alpha s+\alpha Nd}\left( \ee^{-mL_j/2} +L_j^{-2p}\right)\\
&\leq C_3(N,d,|\BK|)L_{k_1}^s + C_2(N,d)L_{k_1}^{Nd}\sum_{j\geq k_1} L_j^{\alpha s+\alpha Nd}\left( \ee^{-\frac{mL_j}{2}} +L_j^{-2p}\right)\\
&< \infty \tag{\text{since $s<\frac{2}{\alpha}p-Nd-1$}}
\end{align*}
\hfill $\qed$

%---------------------------%
%:s.7
%------------------------------------------%
\section{Appendix}\label{sec:appendix}
%------------------------------------------%

%------------------------------------------%
%s.7.1
%------------------------------------------%
%----------------------------------------------------%

%----------------------------------------------------%
%s.7.2
%----------------------------------------------------%
\subsection{Proof of Lemma \ref{lem:separability}}
(A) This assertion is a reformulation of Lemma 1 in \cite{CS09b}
 (see also Lemma 2.1 in \cite{BCS11}), so we omit the proof. \\
(B) Set $R(\By)=\max_{1\leq i,j\leq n}|y_i-y_j|+3NL$ and consider a cube $\BC^{(n)}_L(\Bx)$ with $|\By-\Bx|> R(\By)$. Then there exists $i_0\in\{1,\ldots,n\}$ such that $|y_{i_0}-x_{i_0}|>R(\By)$. Consider the maximal connected component $\Lambda_{\Bx}:=\bigcup_{i\in\CJ} C^{(1)}_{L}(x_i)$ of the union $\bigcup_i C^{(1)}_L(x_i)$ containing $x_{i_0}$. Its diameter is bounded by $2nL$, and by triangle inequality,
\[
\dist(\Lambda_{\Bx},\varPi\BC^{(n)}_L(\By))\geq R(\By)-(\max_{1\leq i,j\leq n}|y_i-y_j| + 2L) - \max_{u,v\in \Lambda_{\Bx}}|u-v|>0,
\]
this implies that $\BC^{(n)}_L(\Bx)$ is $\CJ$-separable from $\BC^{(n)}_L(\By)$ with $\CJ$ the index subset appearing in the definition of $\Lambda_{\Bx}$. $\qed$
%-------------------------------------%
%s.7.3
%-------------------------------------%

\subsection{Proof of Theorem \ref{thm:Wegner}}
The proof extends that of Theorem 2 from \cite{CS08} to an arbitrary $n\geq2$ and is essentially based on Stollmann Lemma (cf. \cite{St01}). Without loss of generality, we can assume that
$\BC^{(n)}(\Bu)$ is pre-separable from $\BC^{(n)}(\Bu')$:
\[
\exists\, \CJ\subset\{1, \ldots, n\}:\quad
\varPi_{\CJ}\BC^{(n)}(\Bu)\cap\left(\varPi_{\CJ^c}\BC^{(n)}(\Bu)\cup\varPi\BC^{(n)}(\Bu')\right)=\emptyset.
\]
(Otherwise, we exchange the roles of $\BC^{(n)}(\Bu)$ and $\BC^{(n)}(\Bu')$.)
Let
\[
\{\lambda^{k}: k=1,\ldots,|\BC^{(n)}(\Bu)|\}, \quad \{\lambda^{k'}: k'=1,\cdots,|\BC^{(n)}(\Bu')|\},
\]
be the eigenvalues of $\BH^{(n)}_{\BC^{(n)}(\Bu)}$ and $\BH^{(n)}_{\BC^{(n)}(\Bu')}$ respectively. Further, let $\FB(\varPi\Bu')$ be the sigma-algebra generated by the random variables
$\{V(y,\cdot),\, y\in \varPi\BC^{(n)}(\Bu')\}$. Then the operator
$\BH^{(n)}_{\BC^{(n)}(\Bu')}$ is $\FB(\varPi\Bu')$-measurable,
thus conditional on $\FB(\varPi\Bu')$, its eigenvalues become non-random.
Therefore, we have:
\begin{equation}\label{eq:prob.esm.cond}
\begin{aligned}
&\prob{\dist(\sigma(\BH^{(n)}_{\BC^{(n)}(\Bu)}),\sigma(\BH^{(n)}_{\BC^{(n)}(\Bu')}))\leq \varepsilon}
\\
& =\esm{ \prob{\dist(\sigma(\BH^{(n)}_{\BC^{(n)}(\Bu)}),\sigma(\BH^{(n)}_{\BC^{(n)}(\Bu')}))\leq \varepsilon \, \Big| \, \FB(\varPi\Bu')}}
\\
& \leq |\BC^{(n)}(\Bu')| \cdot \sup_{\lambda \in\DR }  \; \esm{\prob{\dist(\sigma(\BH^{(n)}_{\BC^{(n)}(\Bu)}),\lambda) \leq \varepsilon
\, \Big| \, \FB(\varPi\Bu')}}\\
& \leq |\BC^{(n)}(\Bu)|\cdot |\BC^{(n)}(\Bu')|\cdot\max_k \cdot \sup_{\lambda \in\DR }  \; \esm{\prob{|\lambda^k_{\BC^{(n)}(\Bu)}(\omega)-\lambda| \leq \varepsilon \, \Big| \, \FB(\varPi\Bu')}}.
\end{aligned}
\end{equation}
For any $k=1,\ldots,|\BC^{(n)}(\Bu)|$, and any $i=1,\ldots,n$,
\begin{equation}\label{eq:prob.esm.cond.2}
\prob{|\lambda^k_{\BC^{(n)}(\Bu)}-\lambda|\leq \varepsilon\, \Big|\,\FB(\varPi\Bu')}=\esm{\prob{|\lambda^k_{\BC^{(n)}(\Bu)}-\lambda|\leq \varepsilon\, \Big|\,\FB(\varPi_{\neq i}\Bu\cup\varPi\Bu')}\, \Big| \, \FB(\varPi\Bu')},
\end{equation}
where $\FB(\varPi_{\neq i}\Bu\cup\varPi\Bu')$ is the sigma-algebra generated by
$\{V(y,\cdot),\, y\in (\varPi\BC^{(n)}(\Bu)\setminus\varPi_i\BC^{(n)}(\Bu))\cup \varPi\BC^{(n)}(\Bu')\}$.  For any $i\in\CJ$, set $J=\varPi_{i}\BC^{(n)}(\Bu)$, $p=|J|$ and denote by $\{y_i:i=1,\ldots,p\}$ the elements of $J$. 
For $v\in\{(V(y_1,\omega),\ldots,V(y_p,\omega)),\omega\in\Omega\}\subset\DR^p$ $\Bx\in\BC^{(n)}(\Bu)$ and $\phi\in\ell^2(\BC^{(n)}(\Bu))$, we define the operators family 
$\BB(v): \phi \mapsto \BB(v)\phi$ by
$$
\BB(v) \phi(\Bx) =\sum_{i=1}^n\sum_{j=1}^p \delta_{x_i,y_j} V(y_j,\omega)\phi(\Bx)=\sum_{j=1}^p \delta_{x_i,y_j} v_j\phi(\Bx).
$$
It is not difficult to see that it satisfies the following two properties in the sense
of quadratic forms:
\begin{enumerate}[\rm(i)]
\item
For all  $r\in\DR^{J}_{+}$, we have
\[
\BB(v+r)\geq \BB(v).
\]
\item
Let $e=e_1+\cdots+e_p\in\DR^p\cong\DR^J$, where $p=|J|$ and $\{e_j\}_{j=1,\ldots,p}$
is the canonical basis of $\DR^J$. Then for all $t \ge 0$
\[
\BB(v+te)-\BB(v)\geq t.
\]
\end{enumerate}
Indeed, item (i) is obvious using the definition of $\BB(v)$. Next, for $v\in\{(V(y_1,\omega),\ldots,V(y_p,\omega)),\omega\in\Omega\}$ by linearity,
\[
((\BB(v+te) - \BB(v))\Phi,\Phi)
=\sum_{\Bx\in\BC^{(n)}(\Bu)}\sum_{i=1}^n\sum_{j=1}^p \delta_{x_i,y_j} t|\Phi(\Bx)|^2
\geq t\|\Phi\|^2.
\]
This proves (ii). 

We will call a parametric family of operators acting in a Hilbert space $\CH$
and indexed by vectors of a Euclidean space $\DR^J$
(with the canonical basis $\{e_i\}$), \emph{diagonally monotone}, if it satisfies
the above properties (i) and (ii).

Note that if $\{\BB(v)\}$ is a diagonally monotone family of operators in a Hilbert space $\CH$, then for any operator $\BK:\CH\rightarrow\CH$, the operator family $v\mapsto \BK+\BB(v)$ is  also diagonally monotone. Thus, since  the operator
$\BH^{(n)}_{\BC^{(n)}(\Bu)}$ admits the decomposition
\[
\BH^{(n)}_{\BC^{(n)}(\Bu)}(\omega) = \sum_{i=1}^n\sum_{j=1}^p \delta_{x_i,y_j}V(y_j,\omega) + \BK(\omega)
\]
where
\[
\BK(\omega) := \BDelta + \BU + \sum_{i=1}^n\sum_{y\in \varPi\BC^{(n)}(\Bu)\setminus J}\delta_{x_i,y} V(y,\omega)
\]
is $\FB(\varPi_{\neq i}\Bu\cup\varPi\Bu')$-measurable, i.e., non-random
conditional on  $\FB(\varPi_{\CJ^{\comp}}\Bu\cup\varPi\Bu')$. As a result,
$\BH^{(n)}_{\BC^{(n)}(\Bu)}$ is a diagonally monotone
family parameterized by $v$. By the min-max principle, each eigenvalue $\lambda^k_{\BC^{(n)}(\Bu)}$, $k=1,\ldots,|\BC^{(n)}(\Bu)|$ veiwed  as a function of $v\in\{(V(y_1,\omega),\ldots,V(y_p,\omega)),\omega\in\Omega\}$ where the random variables $\{V(x,\omega):x\in\varPi\BC^{(n)}(\Bu)\setminus J\}$ are fixed, is also diagonally monotone.
Therefore, Stollmann's lemma (cf. e.g. Lemma 2.1 in \cite{CS08}) applies to each eigenvalue, and we obtain by the independence of the sigma-algebras $\FB(\varPi_{\CJ}\Bu)$ and $\FB(\varPi_{\CJ^{\comp}}\Bu\cup\varPi\Bu')$:

\begin{equation}\label{eq:prob.esm.cond.3}
\begin{aligned}
\prob{|\lambda^k_{\BC^{(n)}(\Bu)}(\omega)-\lambda| \leq \varepsilon
\, \Big| \, \FB(\varPi_{\neq i}\Bu\cup\varPi\Bu')}&\leq\mu^{J}\left\{v: |\lambda-\lambda^k_{\BC^{(n)}(\Bu)}(v)|\leq \varepsilon\right\}
\\
&
\leq |J| \cdot s(F_V,2\varepsilon),
\end{aligned}
\end{equation}
where $\mu^J:=\DP_J$ is the restriction of the probability measure $\DP$ to the sigma-algebra $\FB(\varPi_{\CJ}\Bu)$.
Collecting \eqref{eq:prob.esm.cond}, \eqref{eq:prob.esm.cond.2} and \eqref{eq:prob.esm.cond.3}, we come to the
assertion 
\begin{align*}
&
\prob{\dist(\sigma(\BH^{(n)}_{\BC^{(n)}(\Bu)}),\sigma(\BH^{(n)}_{\BC^{(n)}(\Bu')}))
\leq \varepsilon}
\\
&\qquad
 \leq|\BC^{(n)}(\Bu')|\cdot |\BC^{(n)}(\Bu)| \cdot\max_{i=1,\ldots,n}\max_{\Bu,\Bu'}\{ |\varPi_i\BC^{(n)}(\Bu)|,|\varPi_i\BC^{(n)}(\Bu')|\} \cdot s(F_V,2\varepsilon). \qedhere
\end{align*}

%-------------------%

%---------------------------------------------------------------------%
%:s.acknowledgement
%---------------------------------------------------------------------%
\section*{Acknowledgements}
This work is done in the framework of my PhD Thesis at the Universit\'e Paris Diderot Paris 7. I am  grateful to Anne Boutet de Monvel who suggested me this problem by the year 2009 which marks the begening of the project. In addition, I would  like to thank Victor Chulaevsky for helpful discussions and encouragement. I also thank Mostafa Sabri for reading an earlier version of the text.
%---------------------------------------------------------------------%
%:bib
%---------------------------------------------------------------------%

%---------------------------------------------------------------------%
\end{document}